\documentclass[11pt]{article}
\usepackage{geometry, multirow, amsmath, amsfonts, amsthm, amssymb, graphicx, enumerate, natbib, bibentry, float, caption, subcaption, longtable,xcolor,bm}
\usepackage[breaklinks=true]{hyperref} 
\usepackage[nameinlink]{cleveref} 
\usepackage{xcolor} 
\hypersetup{
	bookmarksnumbered=true,
	colorlinks=true,
	anchorcolor={red!50!black},
	linkcolor={red!50!black},
	citecolor={blue!50!black},
	urlcolor={blue!50!black}
}
\usepackage{txfonts}
\usepackage[title]{appendix}
\usepackage{setspace}

\newenvironment{keywords}{
       \list{}{\advance\topsep by0.35cm\relax\small
       \leftmargin=1cm
       \labelwidth=0.35cm
       \listparindent=0.35cm
       \itemindent\listparindent
       \rightmargin\leftmargin}\item[\hskip\labelsep
                                     \bfseries Keywords:]}
     {\endlist}

\newtheorem{Prop}{Proposition}
\newtheorem{Lem}{Lemma}

\newtheorem{Def}{Definition}

\newenvironment{Defprime}[1]
  {\renewcommand{\theDef}{\ref*{#1}$'$}%
   \addtocounter{Def}{-1}%
   \begin{Def}}
  {\end{Def}}

\let\epsilon\varepsilon

\newcommand\abs[1]{\left|#1 \right|}

\newcommand\Brace[1]{\left\{\begin{aligned}#1\end{aligned}\right.}
\newcommand\p{\partial}

\newcommand\kh[1]{\left(#1\right)}
\newcommand\fkh[1]{\left[#1\right]}
\newcommand\hkh[1]{\left\{#1\right\}}

\newcommand\eqn[1]{\begin{align}#1\end{align}}
\newcommand\eqns[1]{\begin{align*}#1\end{align*}}
\newcommand\indic[1]{\bm{1}\left(#1\right)}

\begin{document}
\setstretch{1.5}
\begin{titlepage}
\title{\textbf{Motivating Effort with Information about Future Rewards}}
\author{Chang Liu\thanks{School of Economics, UNSW Business School; \href{mailto:chang.liu36@unsw.edu.au}{\nolinkurl{chang.liu36@unsw.edu.au}}. I gratefully acknowledge funding from the National Science Foundation under grant DMS-1928930 and from the Alfred P. Sloan Foundation under grant G-2021-16778 during the Fall 2023 semester. I express my deep gratitude to Benjamin Golub and Shengwu Li for their guidance and support throughout the project. 
An earlier version of this paper appeared as the second chapter of my Ph.D. dissertation at Harvard.  
I would also like to thank Yeon-Koo Che, Federico Echenique, Jeffrey Ely, Yannai Gonczarowski, Zo{\"e} Hitzig, Fran\c{c}ois-Xavier Ladant, David Laibson, Annie Liang, Ziqi Lu, John Macke, Erik Madsen, Matthew Rabin, Brit Sharoni, Philipp Strack, Haoqi Tong, Kai Hao Yang, Weijie Zhong, and especially Eric Maskin and Tomasz Strzalecki for valuable comments and discussion. In addition, I thank multiple anonymous referees at ACM EC'22, as well as the audiences at ACM EC'22 and the Stony Brook International Conference on Game Theory for their insightful feedback. 
}}
\date{\begin{tabular}{ rl } 
First version:& October 8, 2021\\
This version:& December 28, 2025
\end{tabular}}
\maketitle
\thispagestyle{empty}
\begin{abstract}
This paper studies the optimal mechanism to motivate effort in a dynamic principal-agent model without transfers. An agent is engaged in a task with uncertain future rewards and can quit at any time. The principal knows the reward and provides information over time to motivate effort. We provide a unified framework and derive the optimal information policy in closed form across stationary and nonstationary environments. Within this framework, we identify two precise conditions, each of which guarantees that dynamic disclosure is strictly valuable. First, if the principal is \textit{impatient} compared to the agent, she prefers the front-loaded effort schedule induced by dynamic disclosure; in a stationary environment, dynamic disclosure is beneficial if and only if the principal is less patient. Second, in an environment where the agent would become \textit{pessimistic} over time absent any disclosure, dynamic information provision can counteract this downward trend and encourage longer effort. Notably, patience remains a crucial determinant of the structure of the optimal policy.
\end{abstract}

\begin{keywords}
Dynamic information design, delayed disclosure, informational incentives
\end{keywords}
\end{titlepage}

\setstretch{1.5}
\newpage\setcounter{page}{2}

\section{Introduction}
In long-term relationships, controlling information is a powerful tool for incentivizing behavior. This paper studies the optimal mechanism to motivate effort in a dynamic principal-agent problem without transfers, focusing on a specific question: When can the principal benefit from revealing information over time rather than all at once?

For instance, consider the relationship between a principal investigator (PI, she) and a postdoctoral researcher (postdoc, he) in a research laboratory. The monetary compensation in a postdoctoral contract usually does not include bonuses for good research results, but the postdoc does build a reputation  that is beneficial to his future career path by achieving major breakthroughs in projects. The PI has worked longer in academia, so she has more information about the prospects and rewards of particular research projects. By revealing information to the postdoc, the PI can influence his attitude towards the project and indirectly motivate his effort. 

How should the PI release information to motivate the postdoc's effort? In particular, is it optimal for her to provide an initial disclosure and then keep silent, or can she benefit from \textit{dynamic disclosure}?  If dynamic disclosure is beneficial, is it better to delay all disclosures to a certain date, or to gradually provide small amounts of information over time? What determines the form of the optimal information policy?\footnote{Our model also applies naturally to master-apprentice relationships (where rewards represent the apprentice's future job prospect) and corporate hierarchies (where rewards correspond to promotions or contract renewals). Moreover, dynamic disclosure can be interpreted as the strategic timing of information-generating events, such as scheduling experiments in scientific research and external committee reviews in business settings.}

We study these questions in a dynamic principal-agent model without transfers (Section \ref{sec:model}). An agent works for a principal on a task and exerts effort continuously until making an \textit{irreversible} decision to quit. The principal aims to maintain the agent's effort, while the agent faces a trade-off between the reward from completing the task and a flow cost of exerting effort. This future reward, which we call the task \textit{quality}, is initially only known by the principal. Task completion is stochastic and arrives according to a Poisson process, and the rate of task completion may depend on the task quality. Both players observe when the task is completed, at which point the relationship ends. The principal can commit to an arbitrary dynamic information policy, specifying how and when to disclose information to the agent about the quality of the task. The agent knows the policy and understands the principal's commitment, observes the realized disclosures, rationally updates his belief about the task quality, and at each moment chooses whether to continue working or to quit in order to maximize his expected payoff.

Our main results provide closed-form characterizations of the optimal information policy and find that  dynamic disclosure is not always necessary: in some environments, a well-chosen static policy is optimal for the principal. However, we identify two precise conditions under which dynamic information disclosure is strictly valuable: one is that the principal is \textit{impatient} compared with the agent, and the other is that the agent would become \textit{pessimistic} over time absent any information disclosure.

We begin with the baseline version of the problem, where the task quality is \textit{binary} and the  environment is \textit{stationary} (Section \ref{sec:bin}). The rate of task completion does not depend on its quality, so the mere passage of time is not informative. We show that the principal benefits from dynamic disclosure if and only if she is \textit{less patient} than the agent (the first part of Proposition \ref{prop:opt}). To see the intuition behind this result, fix some static disclosure policy, and suppose that the principal introduces a short delay before sending any messages. To maintain its persuasive effectiveness, the delayed message must give the agent greater confidence that the task quality is high; without this additional optimism, the agent will find it not worthwhile to wait for disclosure while actively working and incurring flow costs. Overall, the gain from this dynamic adjustment to the static policy is that the agent always starts working until the time of disclosure regardless of the task quality, while the loss is that the agent is more likely to quit from the low-quality task upon disclosure. Because the gain occurs earlier than the loss, the principal prefers this front-loaded effort schedule when she is less patient than the agent.

Moreover, in the binary-stationary environment, the optimal information policy for an impatient principal
takes the form of a \textit{maximally delayed disclosure}: the principal delays all disclosures up to the maximum time threshold and then fully discloses the task quality (the second part of Proposition \ref{prop:opt}). Note that an arbitrary information policy induces a lottery over the agent's quitting time, and this characterization emerges from analyzing how the principal and agent differ in their preferences over such time lotteries. Their incentives align for high quality tasks \textendash{} the optimal policy never induces quitting in this case \textendash{} but diverge for low-quality tasks. When the task quality is low, exponential discounting implies that the principal is risk-averse over time lotteries, while the agent is risk-seeking (since the agent's value is multiplied by a negative constant). The principal's problem reduces to choosing a time lottery that minimizes the loss from quitting (relative to eventual task completion), subject to incentive compatibility constraints. The structure of the optimal time lottery is determined by comparing the risk attitudes of the two parties, which under exponential discounting simplifies to comparing their discount rates. When the principal is less patient than the agent, her greater risk aversion leads her to prefer contracting the time lottery while maintaining the agent's value. This drives the optimal lottery toward minimal risk \textendash{} becoming degenerate at a specific time threshold. Such a lottery corresponds to delayed full disclosure of the task quality. Furthermore, as long as the agent's ex ante value remains strictly positive, the principal benefits from further delay, ultimately leading to maximally delayed disclosure as the optimal policy. 

Next, we analyze the optimal policy when the task quality takes more than two values (Section \ref{sec:gen}), assuming the environment remains stationary and the distribution of task quality has a full-support density. We show that, just as in the binary-stationary case, the principal benefits from dynamic disclosure if and only if she is less patient than the agent. However, the optimal information policy now exhibits richer structure compared with maximally delayed disclosure. It involves a \textit{cutoff cascade} \textendash{} a series of time-varying thresholds where at each instant, the principal reveals only whether the task quality exceeds the current cutoff. Depending on model parameters, this takes one of two forms:  \emph{immediate and gradual disclosure}, which combines immediate information release with gradual disclosure over time, or \emph{delayed gradual disclosure}, which delays all disclosure until a specific time before beginning gradual revelation. Unlike the binary case, even when delay is optimal, the policy involves gradual rather than one-shot disclosure, as the principal now benefits from  customizing disclosure times for different quality realizations. This cascade structure demonstrates the principal's fundamental trade-off between sorting and risk-sharing. On the one hand, efficiency demands sorting effort according to task quality. Earlier disclosure achieves better sorting \textendash{} the agent quits low-quality tasks and continues only with promising ones. On the other hand, if the principal withholds information, the agent remains uninformed and continues working across all quality levels, including low-quality tasks, thereby bearing the risk of wasted effort. Such opacity may increase overall effort provision, but it does so by shifting the risk onto the agent, who may work longer on a low-quality task without knowing it. The optimal information policy balances these two competing forces.

Finally, we move on to nonstationary environments (Section \ref{sec:Poi}) by allowing the
task completion rate to vary with the task quality, assuming the task quality is {binary}. The analysis boils down to two cases, depending on whether the completion rate of the high-quality task is larger than that of the low-quality task. Notably, if the high-quality task has a \emph{shorter} expected duration than the low-quality task, then a new incentive to delay disclosure arises. Here, ``no news is bad news'': the agent becomes  \emph{pessimistic} over time in absence of any information disclosure because it becomes more likely that the task has low quality.  In this pessimistic case, we show that the principal always benefits from dynamic disclosure, \emph{regardless of whether she is more patient than the agent}.\footnote{Conversely, in the case where the agent becomes  \textit{optimistic} over time, the results from the stationary environment directly apply.  If the agent starts working, he will never find it optimal to quit early, so the optimistic case reduces to an essentially stationary problem.} Intuitively, providing information over time can counteract this downward trend in the agent's belief and motivate the agent to work longer. However, we find that the level of patience is still a crucial factor in determining the structure of the optimal information policy (Proposition \ref{prop:optlearn}). Although maximally delayed disclosure remains optimal for an impatient principal, the optimal information policy for a patient principal involves \emph{Poisson disclosure}, where full disclosure arrives at a calibrated Poisson rate in order to restore the stationarity of the environment, i.e., to ensure that ``no news is no news''.

To summarize, our analysis identifies two distinct channels through which dynamic disclosure becomes strictly valuable: (i) relative impatience of the principal, which makes front-loaded effort schedules attractive, and (ii) a pessimistic environment, where dynamic disclosure counteracts the agent's declining beliefs. Table~\ref{tab:summary} organizes the optimal information policies derived in this paper, organized by the quality distribution (binary versus general) and the environment (stationary versus nonstationary).
\begin{table}[!htbp]
\centering
\begin{tabular}{|c|c|c|}
\hline
 & \textbf{Stationary}& \textbf{Nonstationary (Pessimistic)}  \\
\hline
\textbf{Binary} & $\begin{array}{c}\text{Static (if patient)}\\\text{Maximally delayed (if impatient)}\end{array}$ (Section~\ref{sec:bin}) & $\begin{array}{c}\text{Poisson (if patient)}\\\text{Maximally delayed (if impatient)}\end{array}$  (Section~\ref{sec:Poi}) \\
\hline
\textbf{General} &  $\begin{array}{c}\text{Static (if patient)}\\\text{Gradual (if impatient)}\end{array}$ (Section~\ref{sec:gen}) & \textit{Open problem} \\
\hline
\end{tabular}
\caption{Summary of optimal information policies across environments.}
\label{tab:summary}
\end{table}
The general nonstationary case remains an open problem for future research.

Beyond identifying when dynamic disclosure is strictly valuable, our framework provides a complete characterization of optimal policies, revealing how the solution structure undergoes phase transitions as underlying parameters change. By varying players' discount rates, the reward distribution, or task completion rates, the optimal policy shifts between distinct patterns that connect to existing literature. In some regions, the solution echoes \citet{ES20} through maximally delayed disclosure that front-loads effort. In others, it resembles \citet{SmolinFC} through gradual, state-contingent revelation managing the trade-off between sorting efficiency and risk allocation. While these earlier contributions each feature one type of mechanism, our unified framework shows that these emerge as different phases of a single model, with boundaries determined by threshold conditions on patience and the environment's informational properties. 

What drives these phase transitions is which incentive constraints bind at optimality. In stationary environments, only the agent's ex ante individual rationality constraint binds at optimality, and the players' relative patience determines whether the principal uses static or dynamic disclosure through risk-shifting considerations arising from exponential discounting. In nonstationary environments, the challenge becomes identifying which of the agent's many incentive constraints across different times must bind. In the pessimistic case, the incentive constraint at the moment the agent's belief decays to the quitting threshold becomes crucial (our Free-Riding Lemma~\ref{lem:free}), and multiple dynamic obedience constraints may bind simultaneously. This shift in binding constraints \textendash{} from static individual rationality to dynamic obedience \textendash{} explains why the optimal policies range from simple delay to complex gradual rules.

On the technical side, we develop a unified proof approach that handles both stationary and nonstationary environments. Throughout, we employ a change-of-variables technique that transforms the principal's problem into an optimization over time lotteries, where comparing risk attitudes under exponential discounting yields closed-form solutions.\footnote{We thank an anonymous referee for suggesting this change of variables for the stationary case.} For stationary environments, this works cleanly because dynamic incentive constraints reduce to a single ex ante condition. Our main technical contribution is extending this approach to nonstationary environments, where the natural drift of beliefs introduces a continuum of interim constraints that must be handled simultaneously.

\subsection{Related Literature}
The literature on information design and Bayesian persuasion, to which this paper contributes, was pioneered by the work of \citet{AMS95} and \citet{KG11}. Within this broad category, there is a particularly relevant literature on dynamic information design, which focuses on the role of information control in dynamic environments in terms of motivating and coordinating actions.

Our model is closest to \citet{ES20}, which also studies the optimal information policy for a principal to motivate effort in a dynamic moral hazard framework without transfers. In \citet{ES20}, the asymmetric information is about the duration of the task: the principal can conceal task completion, strategically keeping the agent working even after the task is complete. Because the principal can always exploit this hidden information, dynamic disclosure is universally optimal. Our model, by contrast, assumes both parties observe task completion, precluding such exploitation. Dynamic disclosure then need not outperform static disclosure, and our unified framework characterizes precisely when it does: if and only if one of the two channels we identify is present. 

A second closely related paper is \citet{SmolinFC}, which studies optimal performance evaluation design in a career-concerns model where an agent of uncertain productivity decides when to quit based on feedback. In \citet{SmolinFC}'s setting, the agent learns from evaluations derived from stochastic performance signals, and the principal designs how these signals are aggregated into evaluations. The equilibrium policy produces a deterministic and typically increasing wage-tenure profile through sorting: agents who remain are increasingly likely to be high types. When dynamic disclosure is optimal in our general stationary case, the policy structure is similar \textendash{} a cascade of increasing cutoffs over time \textendash{} but the underlying mechanism differs. In our setting, the structure stems not only from sorting but also from front-loading effort, which benefits only an impatient principal.

This paper connects to \citet{Au15} and \citet{GS18}, as it identifies conditions under which the optimal disclosure plan in a dynamic framework is static. However, their setting involves dynamic persuasion of a privately informed receiver, where the sender reveals information sequentially to screen the receiver's type. Also related are \citet{Ely17}, \citet{RSV17} and \citet{Ball19}; in their settings, the principal privately observes a state that evolves according to some exogenous stochastic process, and the structure of the optimal information policy depends on this process. By contrast, the model in this paper characterizes the timing and form of information disclosure when the state is fixed. Moreover, in both  \citet{Ely17} and \citet{RSV17} the agent acts myopically, whereas this paper studies the dynamic relationship between two long-lived players.  

Several other recent papers study related topics but differ in their assumptions about what information the principal can control. In \citet{Kaya20}, the principal can commit within each period to any signal structure but cannot commit over time to a sequence of signals.\footnote{The commitment assumption along this style also appears in \citet{OSZ20}, \citet{BRV21}, and \citet{ES21}.} Thus, the relevant solution concept is perfect Bayesian equilibrium. In our model, the principal has full commitment power over arbitrary signal processes, yet our result that a static policy can be dynamically optimal indicates that this assumption is stronger than necessary in some cases. \citet{OSZ20} and \citet{BRV21} assume instead that the principal is not the only source of information: in their models, there is an exogenous flow of payoff-relevant information beyond the principal's control. Furthermore, \citet{CKM21} introduce costs of generating and processing information, and obtain a version of the folk theorem as these costs vanishes.\footnote{Recent working papers, circulated after earlier versions of the present work, explore related themes in dynamic information design. For example, \cite{KSZ24} develop a unified framework for persuasion with optimal stopping, establishing a revelation principle under full commitment and characterizing implementable outcomes without intertemporal commitment. Our focus on the conditions under which dynamic disclosure is strictly optimal is complementary.}
\\ \\ 
\indent The rest of the paper is organized as follows. Section \ref{sec:model} lays out the model. The first main part, Section \ref{sec:bin}, analyzes the binary-stationary case and derives one condition that makes dynamic disclosure valuable: the principal is \textit{impatient} compared with the agent. The second main part, Section \ref{sec:gen}, extends the analysis to general reward distributions and reinforces the results from the binary case. The third main part, Section \ref{sec:Poi}, analyzes the nonstationary environments and reveals another channel through which dynamic disclosure becomes valuable: the agent becomes \textit{pessimistic} over time without any information disclosure. Section \ref{sec:concln} concludes. Appendix \ref{app:proof} contains the proofs of all results in the main text.

\section{Model}\label{sec:model}
The model is a continuous-time principal-agent problem without transfers. An agent works for a principal and exerts effort until he chooses to quit. We assume that the quit decision is \textit{irreversible}. The agent exerts effort at flow cost $c>0$, earns a reward $q\ge 0$ if he works until the task is completed, and discounts at rate $r>0$. We refer to the reward, $q$, as the quality of the task. The realized quality $q$ is the principal's private information, and the agent's prior is given by a distribution $F$ with finite first moment: $\mathbb{E}_F\fkh{q}<\infty$.

Assume that both players observe when the task is completed, and the completion arrives with Poisson rate $\lambda_q>0$ when the task quality is $q$, at which point the relationship ends. This implies that when the task quality is $q$, the task duration $x$ is an exponential random variable with mean $1/\lambda_q$. Thus, given realized duration $x$, the agent's ex post payoff from exerting effort until time $\tau$ is
$$e^{-r\tau}q\cdot \indic{\tau\ge x}-\int_0^\tau e^{-rt}c\,dt=e^{-r\tau}q\cdot \indic{\tau\ge x}-\frac{c}{r}\left(1-e^{-r \tau}\right).$$
The principal receives a bonus $b\kh{q}>0$ if the agent works until the task is completed, and discounts at rate $r_p>0$; therefore, her ex post payoff when the agent quits at $\tau$ is $$e^{-r_p\tau}b\kh{q}\cdot \indic{\tau\ge x}.$$
We assume $b(\cdot)$ is a bounded measurable function with $0<\underline{b} \leq b(q) \leq \overline{b}<\infty$ for all $q\in[0,\infty)$. These assumptions ensure that both players' expected payoffs are always well-defined.

We now compute the following expected payoff functions, which will be useful for later analysis. The agent's expected payoff from completing a task with quality $q$ (i.e., quit at $\tau=\infty$) is
\begin{align*}
v(q) &\equiv \mathbb{E}_{x\mid q}\left[ e^{-r x} q-\frac{c}{r}\left(1-e^{-r x}\right) \,\middle|\, q \right] \\
&= \int_{0}^{\infty} \left[ e^{-rx}q-\frac{c}{r}\left(1-e^{-r x}\right) \right] \lambda_q e^{-\lambda_q x}\,dx= \frac{\lambda_q q-c}{r+\lambda_q}.
\end{align*}
The agent may also plan to quit at some time $\tau$ if the task is still incomplete. This leads to expected payoff
\begin{align}
v(\tau,q) &\equiv \mathbb{E}_{x\mid q}\left[ e^{-rx}q \cdot \indic{x\le \tau} - \frac{c}{r}\left(1-e^{-r \min\{x,\tau\}}\right) \,\middle|\, q \right] \notag \\
&= \frac{\lambda_q q-c}{r+\lambda_q}\left(1-e^{-(r+\lambda_q) \tau}\right) = \left(1-e^{-(r+\lambda_q) \tau}\right) v(q). \label{eqn:v}
\end{align}
Similarly, the principal's expected payoff when the agent chooses to quit at $\tau$ is
\begin{align}
w(\tau,q) &\equiv \mathbb{E}_{x\mid q}\left[ e^{-r_px} b(q) \cdot \indic{x\le \tau} \,\middle|\, q \right] \notag \\
&= \frac{\lambda_q b(q)}{r_p+\lambda_q}\left(1-e^{-(r_{p}+\lambda_q) \tau}\right) \equiv \left(1-e^{-(r_{p}+\lambda_q) \tau}\right) w(q). \label{eqn:w}
\end{align}
In particular, $w(q)=\frac{\lambda_q b(q)}{r_p+\lambda_q}$ denotes the principal's expected payoff from a completed task with quality $q$.

Assume that the principal can commit to an arbitrary rule that specifies how and when to disclose information to the agent about the quality of the task.\footnote{In organizations, commitment can be sustained formally through contracts (e.g., reviewed by an external evaluator), or informally through reputation. Moreover, commitment can also be understood as using existing instruments to generate signals that cannot be hidden. In the previous example of a research lab, suppose the project is to develop a new empirical method, and the PI has control over exclusive data sets that can improve the understanding of the value of completing a breakthrough through test runs. The PI can decide which data set to provide to the postdoc and when to provide it, but once the data set is provided, the results of the test run will be publicly observed. Finally, as we show in the following sections, when the principal is patient and the environment either is stationary or makes the agent become optimistic over time, a static information policy turns out to be dynamically optimal. In these cases, full commitment power over arbitrary signal processes is stronger than necessary.} She can make any number of disclosures at any time during the process. The agent knows the rule and understands the principal's commitment, observes the realized disclosures, rationally updates his belief about the quality of the task, and best responds by making effort choices that maximize his expected payoff. Adapting the proof of Proposition 1 in \citet{KG11}, one can show that it is without loss of generality to restrict attention to a particular class of rules that produce ``recommended actions'', an analog to the revelation principle. Since the agent's decision to quit is irreversible, the only relevant instruction the principal can provide is the time at which the agent should stop. We refer to a rule governing these disclosures as an \emph{information policy}, which informs the agent whether to continue ($\sigma_t=0$) or quit ($\sigma_t=1$) at any time. 
\begin{Def}[Information policy]
An \emph{information policy} $\bm\sigma=\{\sigma_t\}_{t\ge 0}$ is a $\{0,1\}$-valued, non-decreasing, and right-continuous stochastic process whose distribution depends on the task quality $q$. The policy is fully characterized by the family of conditional probabilities $\mathbb{P}(\sigma_t = 1 \mid q)$, which represents the probability that quitting has been recommended by time $t$ given quality $q$.
\end{Def}
The principal's objective  is to maximize her expected payoff $\mathbb{E}_{\bm{\sigma}}[w(\tau,q)]$ over all possible information policies $\bm\sigma$, where $\tau$ denotes the agent's random quitting time induced by $\bm\sigma$.\footnote{The agent's quitting time is determined by comparing two payoffs at each instant of time: the continuation payoff from following the principal's recommendations versus the value of ignoring further information. We omit the formal expressions from the main text for brevity.}

In the following sections, we derive the optimal information policy in closed form, in order to obtain the necessary and sufficient conditions for the principal to benefit from dynamic disclosure.

\section{Binary-Stationary Case}\label{sec:bin}

We begin our analysis with the baseline version of the problem, where the task quality takes one of two possible values: low quality, $L>0$, or high quality, $H>L$, and the environment is stationary: $\lambda_{H}=\lambda_{L}=\lambda$. The agent's prior belief assigns probability $\mu\in\kh{0,1}$ that the task has high quality.

Assume $v\kh{H}>0>v\kh{L}$, so that the low-quality task is not individually rational for the agent.
We fully characterize the optimal dynamic information policy, and, through this analysis, reveal the first channel that makes dynamic information disclosure valuable: the principal is \emph{impatient} compared with the agent.

The main result for the binary-stationary case is Proposition \ref{prop:opt}, which shows that the principal benefits from dynamic disclosure if and only if she is \emph{less patient}\footnote{While it is common to assume that the principal is more patient than the agent, we want to point out that the principal may be relatively impatient in some cases. In the previous example of a research lab, the PI could be under pressure for a promotion review, and the postdoc may have just been hired.} than the agent, i.e., $r_p>r$. Moreover, if $r_p>r$, \emph{maximally delayed disclosure} is optimal: the principal delays all disclosure up to the maximum time and then fully discloses the task quality.

\subsection{Static Disclosure}
Consider a one-shot information release at time zero, after which the agent decides how long to work.  Under this static disclosure, once the agent receives the principal's initial message, he effectively faces a binary choice: continue working on the task until completion or quit immediately. This is because the environment is time-stationary: by the memoryless property of Poisson arrivals, the agent's continuation problem at any time $t>0$ is identical to the one at time zero. Hence, the static version reduces to the standard Bayesian persuasion problem of \citet{KG11}.

Let $\overline{\mu}$ denote the threshold belief at which the agent is indifferent between working and quitting:
\eqn{\overline{\mu} v\kh{H}+\kh{1-\overline{\mu}}v\kh{L}=0.\label{eqn:mubar}}
For $\mu\ge \overline{\mu}$, no persuasion is necessary. When $\mu<\overline{\mu}$, the principal's objective is to design an initial signal that maximizes the probability the agent's belief remains above the $\overline{\mu}$  threshold, thereby inducing him to work. The well-known solution is a binary signal that we refer to as the $\operatorname{KG}$ policy.
\begin{Def}\label{def:KG}
$\operatorname{KG}$ denotes the static information policy that,  for $\mu<\overline{\mu}$, satisfies:
\begin{enumerate}
\item The agent is told to continue ($\sigma_0=0$) with probability one when the task has high quality ($q=H$), and also with positive probability when the task has low quality ($q=L$).
\item The agent exactly holds belief $\overline{\mu}$ conditional on observing $\sigma_0=0$, i.e., $\mathbb{P}(q=H \mid \sigma_0=0)=\overline{\mu}$.
\item The policy is static: for all $t > 0$, $\sigma_t = \sigma_0$ almost surely.
\end{enumerate}
\end{Def}
The principal's value function under KG with respect to the agent's belief, $W^{\operatorname{KG}}(\mu)$, is depicted in Figure \ref{fig:WKG}. The solid line segments represent the benchmark without persuasion, and the dashed line segment, formed by the concavification of the solid line segments, represents the principal's gain obtained from persuasion.\footnote{Only for illustration, Figure \ref{fig:WKG}, as well as Figures \ref{fig:DD(t)vs.KG}, \ref{fig:WFDD}, and \ref{fig:WFDD1} below, assumes $b(L)=b(H)$; this does not affect any of our results.}
\begin{figure}[!htbp]\centering
\includegraphics[height=0.3\textheight]{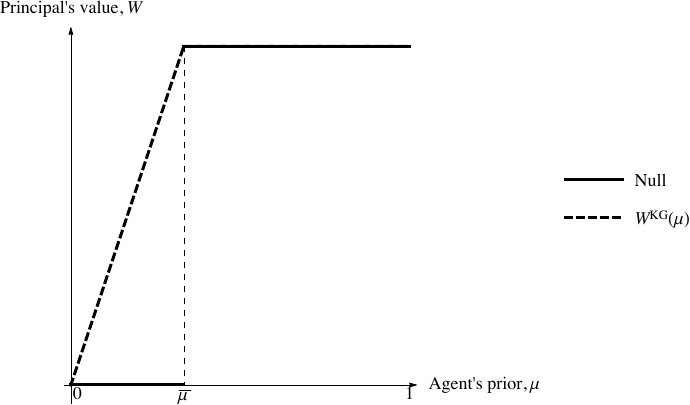}
\caption{Optimal static information policy $\operatorname{KG}$ corresponds to a concavification.}\label{fig:WKG}
\end{figure}

The static $\operatorname{KG}$ policy serves as a benchmark for the dynamic policies we consider next. In fact, we will see that if the principal is relatively patient (low $r_p$), this static disclosure remains dynamically optimal even when more complex information policies  are available.

\subsection{Delayed Disclosure}
We now allow the principal to spread out information over time and examine whether doing so can lead to higher payoffs. Consider the simplest dynamic policy, where the principal withholds all information until time $t>0$, then fully discloses the task quality.\footnote{Of course, the principal could employ more complex dynamic policies that do not fully reveal the state, but as we later show, such complexity is suboptimal in the binary-stationary case.} Again, due to the memoryless property, conditional on working until time $t$, the agent will quit when informed of low quality and continue until completion otherwise. We call this policy \emph{delayed disclosure} at time $t$, denoted by $\operatorname{DD}(t)$.
\begin{Def}[Delayed disclosure]\label{def:DD(t)}$\operatorname{DD}(t)$ denotes the information policy of \emph{delayed disclosure} at time $t$, where the principal tells the agent nothing prior to time $t$, and fully discloses the task quality at $t$. Formally, $\sigma_s = 0$ for all $s < t$, and for all $s \ge t$, $\sigma_s = \indic{q = L}$.\end{Def}
Of course, delaying information can backfire if the agent's initial belief is too low \textendash{} the agent might refuse to work from the outset unless he anticipates a sufficient chance of a high-quality task.  Let $\tilde{\mu}(t)$ denote the minimum prior belief at which the agent is willing to start working under the policy $\text{DD}(t)$. This $\tilde{\mu}(t)$ is determined by the agent's ex ante incentive constraint: the solution to\footnote{Since $\overline{\mu}$ is the smallest belief at which the agent would start working under no disclosure, we have $\tilde {\mu}\kh{t}<\overline{\mu}$ for all $t>0$. It follows from equation \eqref{eqn:v} that 
$v\kh{\tau,L}=\left(1-e^{-(r+\lambda) \tau}\right) v(L)$ with $v\kh{L}<0$, so $v\kh{\tau,L}$ strictly decreases from $0$ to $v\kh{L}$ as $\tau$ increases from $0$ to infinity. Comparing equations \eqref{eqn:mubar} and \eqref{eqn:mutilde} that define the threshold priors $\overline{\mu}$ and $\tilde{\mu}(t)$, we can see that (i) $\tilde {\mu}\kh{t}$ is strictly increasing in $t$, (ii) $\lim_{t\to 0}\tilde{\mu}\kh{t}=0$, $\lim_{t\to\infty}\tilde {\mu}\kh{t}=\overline{\mu}$.} 
\eqn{\tilde{\mu}\kh{t} v\kh{H}+\kh{1-\tilde{\mu}\kh{t} }v(t,L)=0.\label{eqn:mutilde}}

Next, we examine whether $\operatorname{DD}(t)$ enables the principal to increase her payoff relative to $\operatorname{KG}$. The answer depends crucially on patience, i.e., the relationship between $r_p$ and $r$. The gain from $\operatorname{DD}(t)$ compared with $\operatorname{KG}$ is that the agent always starts working regardless of the task quality, and the loss from $\operatorname{DD}(t)$ is that the agent quits from the low-quality task after the disclosure at time $t$. Since the gain occurs earlier than the loss, an impatient principal benefits from $\operatorname{DD}(t)$. The formal result is summarized in the following lemma.
\begin{Lem}\label{lem:1}
Fix a time $t>0$. $\operatorname{DD}(t)$ increases the principal's payoff relative to $\operatorname{KG}$ for some range of priors if and only if the principal is less patient than the agent, i.e., $r_p>r$.
\end{Lem}
\begin{proof}
All proofs of the results in the main text are in Appendix \ref{app:proof}.
\end{proof}
Figure \ref{fig:DD(t)vs.KG} illustrates this pattern: when $r_p>r$, there exists a range of priors for which $\operatorname{DD}(t)$ yields higher payoffs than $\operatorname{KG}$ (Figure \ref{fig:imp}); when $r_p<r$, $\operatorname{DD}(t)$ is worse for the principal than $\operatorname{KG}$ for all priors (Figure \ref{fig:pat}). 
\begin{figure}[!htbp]
\begin{subfigure}[t]{0.49\textwidth}\centering
\includegraphics[width=\textwidth]{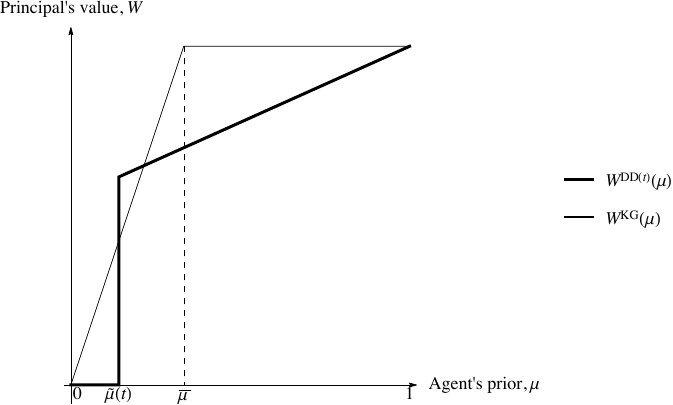}
\caption{$r_p>r$.}\label{fig:imp}
\end{subfigure}\hfill
\begin{subfigure}[t]{0.49\textwidth}\centering
\includegraphics[width=\textwidth]{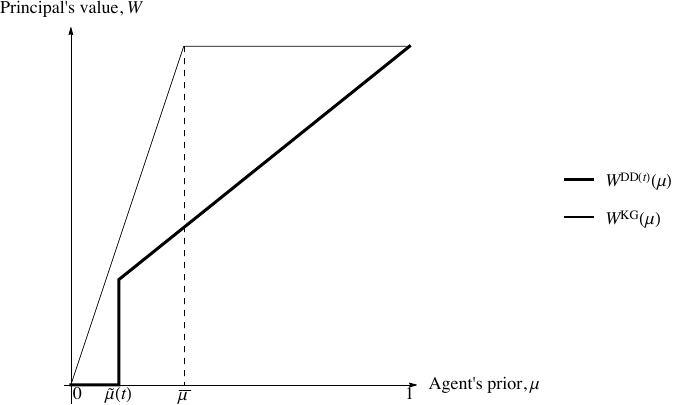}
\caption{$r_p<r$.}\label{fig:pat}
\end{subfigure}
\caption{The principal's payoff comparison between $\operatorname{DD}(t)$ and $\operatorname{KG}$.}\label{fig:DD(t)vs.KG}
\end{figure}

Taking one step further, we argue that whenever the principal gains from delaying disclosure (i.e., when $r_p>r$), she can improve her payoff by delaying as much as possible subject to keeping the agent on board. For any $\mu<\overline{\mu}$, let $\tilde{t}\kh{\mu}$ denote the \emph{maximum delay time}, i.e., the latest time of disclosure subject to the constraint that the agent's ex ante payoff is at least zero:\footnote{Note from equations \eqref{eqn:mutilde} and \eqref{eqn:ttilde} that $\tilde{t}\kh{\mu}$ is the inverse of $\tilde{\mu}\kh{t}$, so it satisfies (i)  $\tilde {t}\kh{\mu}$ is strictly increasing in $\mu$ within the range $\kh{0,\overline{\mu}}$, (ii) $\tilde{t}\kh{0}=0$, $\lim_{\mu\uparrow\overline{\mu}}\tilde {t}\kh{\mu}=\infty$, (iii) $\tilde {t}\kh{\mu}=\infty$ for all $\mu\ge \overline{\mu}$.}
\eqn{\mu v(H)+(1-\mu) v\kh{\tilde{t}(\mu), L}=0.\label{eqn:ttilde}}
Among all disclosure times $t$ that are sufficient to motive the agent under $\operatorname{DD}\kh{t}$ (i.e., $t\le \tilde{t}(\mu)$), the principal's payoff is increasing in $t$. Intuitively, a longer wait means the agent will work longer on the low-quality task before quitting, which is precisely what the impatient principal values. Therefore, the optimal choice is to set $t = \tilde{t}(\mu)$, the longest delay that just barely induces the agent to participate. We refer to this extreme policy as  \emph{maximally delayed disclosure}, denoted by $\operatorname{MDD}$.
\begin{Def}[Maximally delayed disclosure]\label{def:FDD}
$\operatorname{MDD}$ denotes the information policy of \emph{maximally delayed disclosure}. For $\mu<\overline{\mu}$, the principal provides no information prior to the time $\tilde{t}(\mu)$, and fully discloses the task quality at $\tilde{t}(\mu)$. Formally, $\sigma_s = 0$ for all $s < \tilde{t}(\mu)$, and for all $s \ge \tilde{t}(\mu)$, $\sigma_s = \indic{q=L}$.
\end{Def}
By construction, for $\mu<\overline{\mu}$, the agent is exactly indifferent to starting work under $\operatorname{MDD}$ (he gets zero expected payoff), so he complies with the recommendation to continue until $\tilde {t}\kh{\mu}$. Figure \ref{fig:WFDD} depicts the principal's value under $\operatorname{MDD}$ compared to simpler $\operatorname{DD}(t)$ policies and the static $\operatorname{KG}$ policy when $r_p>r$.
\begin{figure}[!htbp]
\begin{subfigure}[t]{0.49\textwidth}\centering
\includegraphics[width=\textwidth]{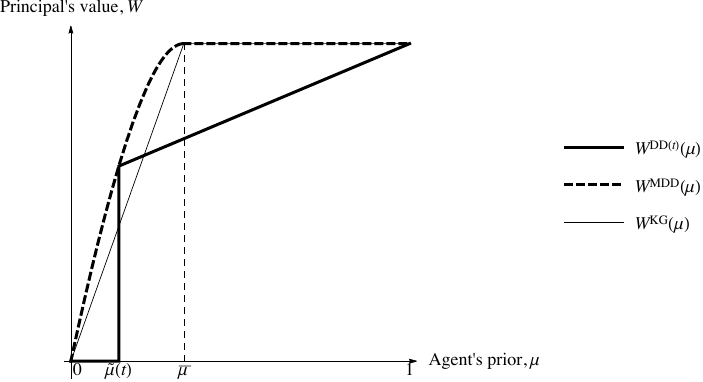}
\caption{$r_p>r$.}\label{fig:WFDD}
\end{subfigure}\hfill
\begin{subfigure}[t]{0.49\textwidth}\centering
\includegraphics[width=\textwidth]{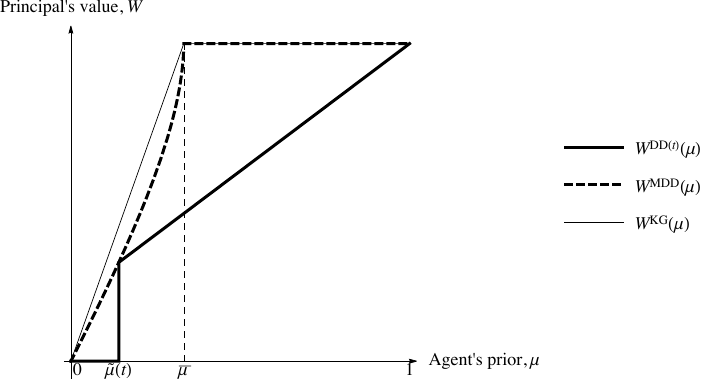}
\caption{$r_p<r$.}\label{fig:WFDD1}
\end{subfigure}
\caption{The principal's payoff comparison among $\operatorname{DD}(t)$, $\operatorname{MDD}$ and $\operatorname{KG}$.}
\end{figure}
In particular, for all $\mu\in\kh{0,\overline{\mu}}$, $\operatorname{MDD}$ strictly increases the principal's payoff relative to the optimal static policy $\operatorname{KG}$. Moreover, $W^{\operatorname{MDD}}$ is the upper envelope and concave closure of the family of value functions $\hkh{W^{\operatorname{DD}(t)}:t>0}$. Hence, $W^{\operatorname{MDD}}$ itself is a concave function, and it follows that $\operatorname{MDD}$ yields a higher payoff than any information policy with a less informative disclosure. In other words, once we allow the disclosure time to be optimally chosen as a function of the prior, gradual or partial disclosures cannot outperform the bang-bang approach of $\operatorname{MDD}$. In Proposition \ref{prop:opt}, we establish that $\operatorname{MDD}$ is indeed optimal within the set of all dynamic information policies.

Conversely, consider the case $r_p<r$. The relationship among the information policies $\operatorname{DD}(t)$, $\operatorname{MDD}$ and $\operatorname{KG}$ is shown in Figure \ref{fig:WFDD1}. The principal still prefers $\operatorname{MDD}$ to any policy in the family $\hkh{\operatorname{DD}(t):t>0}$, but $W^{\operatorname{MDD}}$ is no longer concave, and its concavification is precisely $W^{\operatorname{KG}}$. In this case, Proposition \ref{prop:opt} establishes the optimality of the static policy  $\operatorname{KG}$ within all dynamic information policies. In particular, when the principal is more patient than the agent, her commitment power over the entire signal process has no additional value.

\subsection{Optimal Information Policy}\label{subsec:opt} 
We are now ready to state the main result for the binary-stationary case, Proposition \ref{prop:opt},which establishes the optimality of the information policies analyzed above. Among all dynamic information policies, KG is optimal when $r_p\le r$, and MDD is optimal when $r_p>r$.\footnote{When $r_p=r$, a range of policies including KG and MDD are optimal, as long as they are sufficient to motivate the agent and such that the agent's ex ante expected payoff is exactly zero (e.g., non-maximally delayed partial disclosure).}
\begin{Prop}\label{prop:opt}In the binary-stationary case,
\begin{enumerate}
\item If $r_p\le r$,  then $\operatorname{KG}$ is optimal among all dynamic information policies.
\item If $r_p>r$, then $\operatorname{MDD}$ is optimal among all dynamic information policies.
\end{enumerate}	
\end{Prop}
The optimality is within the entire set of dynamic information policies, including those with multiple or even infinite disclosures that may induce quitting times other than zero or $\tilde{t}\kh{\mu}$.

The intuition behind Proposition \ref{prop:opt} is best understood by viewing the information design problem as choosing a lottery over the agent's quitting time. Any information policy, $\bm{\sigma}$, induces a distribution over times at which the agent quits, $\tau$. The principal wants to make this quitting time as ``late'' as possible, but at the same time she must respect the agent's incentive compatibility constraints. Because the agent never needs to be encouraged to continue when the task quality is high, the non-trivial design problem only concerns when the task quality is low. 

Crucially, under exponential discounting, the principal and agent have opposing risk preferences over time lotteries: the principal is risk-averse, while the agent is risk-seeking.\footnote{The connection between exponential discounting and risk preferences over time lotteries already features in the analysis in Section IV.F of \cite{ES20}.} Formally, applying Jensen's inequality to equations \eqref{eqn:v} and  \eqref{eqn:w}:
\eqns{&\mathbb{E}_{\bm{\sigma}}[w(\tau,L)]=w\kh{L}\cdot \mathbb{E}_{\bm{\sigma}}\left[1-e^{-\left(r_{p}+\lambda\right) \tau}\right] \leq w\kh{L}\left(1-e^{-\left(r_{p}+\lambda\right) \mathbb{E}_{\bm{\sigma}}[\tau]}\right)=w\kh{\mathbb{E}_{\bm{\sigma}}[\tau],L}, \\ 
&\mathbb{E}_{\bm{\sigma}}[v(\tau, L)]=\underbrace{v(L)}_{(-)} \cdot\, \mathbb{E}_{\bm{\sigma}}\left[1-e^{-(r+\lambda) \tau}\right] \geq \underbrace{v(L)}_{(-)} \cdot\, \left(1-e^{-\left(r+\lambda\right) \mathbb{E}_{\bm{\sigma}}[\tau]}\right)=v(\mathbb{E}_{\bm{\sigma}}[\tau], L).}
Given this difference, the optimal policy depends on whose risk attitude dominates, and skews towards one of two extremes depending on who is more patient.  If $r_p<r$, the agent's risk-seeking preference dominates, favoring an extreme lottery with maximum uncertainty: either immediate quitting or indefinite continuation, precisely what $\operatorname{KG}$ delivers. On the other hand, if $r_p>r$,  the principal's risk-aversion dominates, favoring a deterministic quitting time, which $\operatorname{MDD}$ provides by concentrating all probability mass at $\tilde{t}\kh{\mu}$.

\section{General Stationary Case}\label{sec:gen}
We now characterize the optimal dynamic information policy for general reward distributions $F$, assuming that the environment is stationary ($\lambda_{q} \equiv \lambda$) and that $F$ has a full-support density over $[0,\infty)$. The latter assumption implies that it is without loss of optimality to restrict attention to policies that induce a deterministic quitting time $\tau(q)$ for each quality realization.\footnote{In discrete settings (e.g., Sections \ref{sec:bin} and \ref{sec:Poi}), optimal policies often require randomization to convexify the principal's feasible payoffs and satisfy individual rationality constraints with equality. With an atomless quality distribution, the principal can bind these constraints by continuously adjusting the quality cutoff, thereby eliminating the need for randomization.} In addition, we assume that the agent's ex ante payoff from always completing the task is negative, i.e., $\mathbb{E}_{q}[v(q)]=\int_{0}^{\infty} v(q)\, d F(q)<0$, as otherwise no persuasion would be necessary. Finally, we assume that the principal's bonus $b\kh{q}$ is nondecreasing in task quality $q$.\footnote{This assumption applies to many situations. One example is that the principal receives the same bonus at completion regardless of the task quality ($b\kh{q}=b_0$). Another is that the principal and the agent have the same reward from the task ($b\kh{q}=q$), so the difference between their payoffs is only that the principal does not incur the costs. Furthermore, the results for this section would go through even if $b\kh{q}$ may be decreasing, as long as it does not decrease too rapidly. Formally, we need the ratio $\abs{v\kh{q}}/b\kh{q}$ to be nonincreasing over $q\in\kh{0,\overline{q}}$.}

The main result for this section parallels the binary-case finding from Section \ref{sec:bin}: in a stationary environment, the principal benefits from dynamic disclosure if and only if she is less patient than the agent ($r_p>r$). If $r_p\le r$, Proposition \ref{prop:gKG} shows that the optimal information policy is an adaptation of $\operatorname{KG}$. By contrast, if $r_p>r$, the principal strictly benefits from dynamic disclosure, and the optimal policy becomes a \emph{cutoff cascade} \textendash{} a series of time-varying thresholds where at each instant, the principal reveals only whether the task quality exceeds the current cutoff, as shown in Propositions \ref{prop:gstar} and \ref{prop:gFDD}.  This gradualism custom-fits timing to realized quality, generalizing $\operatorname{MDD}$ from a single cutoff to a filtering of low qualities over time.

The section proceeds as follows. Subsection~\ref{sec:static_gen} establishes the optimality of static disclosure when the principal is patient ($r_p\le r$). Subsection~\ref{sec:cutoff} shows that when the principal is impatient ($r_p>r$), the optimal policy takes the form of a \textit{cutoff cascade}. Subsections~\ref{sec:IGD} and \ref{sec:DGD} then characterize the specific structure of  the optimal cascade.

\subsection{Static Disclosure}\label{sec:static_gen}
As in the binary case, we first consider the static version where the principal provides a one-shot disclosure and the agent chooses effort. Because $F$ has full support and no atoms, there exists a unique task quality $q^*>0$ such that \eqn{\int_{q^*}^{\infty}v\kh{q}\,dF\kh{q}=0.\label{eqn:qstar}} Let $\overline q$ denote the \emph{minimum individually rational quality} for the agent; that is, $\overline q$ solves $$v\kh{\overline{q}}=\frac{\lambda \overline{q}-c}{r+\lambda}=0,\quad\text{i.e.,}\quad\overline{q}=\frac{c}{\lambda}.$$
It follows from equation \eqref{eqn:qstar} that $q^*<\overline{q}$.

Consider the following static information policy, which is a counterpart of KG: the principal informs the agent whether the task has high quality ($q\ge q^*$) or low quality ($q<q^*$). By stationarity, if informed of high quality, the agent completes the task ($\tau=\infty$); otherwise, he quits immediately  ($\tau=0$). Slightly abusing notation, we also use $\operatorname{KG}$ to represent this information policy.
\begin{Def}
$\operatorname{KG}$ denotes the static information policy that informs the agent whether the task has high quality ($q\ge q^*$) or low quality ($q<q^*$). Formally, $\sigma_0\kh{q}=\indic{q < q^*}$ and $\sigma_t = \sigma_0$ for all $t \ge 0$, with induced quitting time
$$\tau^{\operatorname{KG}}\kh{q}=\Brace{&0,&&q<q^*,\\&\infty,&&q\ge q^*.}$$
\end{Def}
Indeed, the initial disclosure motivates the agent to complete some tasks with lower quality than the minimum individually rational level (those in the range $q^*\le q<\overline{q}$). As an analog of the first part of Proposition \ref{prop:opt}, we show below that if $r_p\le r$, this static policy is optimal among all dynamic ones.
\begin{Prop}\label{prop:gKG}
If $r_p\le r$,	then $\operatorname{KG}$ is optimal among all dynamic information policies.
\end{Prop}

\subsection{Cutoff Cascades}\label{sec:cutoff}
Henceforth, we consider the case where the principal is less patient than the agent, $r_p>r$. As in the binary case, the principal strictly benefits from dynamic disclosure over static disclosure. However, with a continuum of qualities, the optimal dynamic policy differs from $\operatorname{MDD}$: rather than a single delayed revelation, the principal now customizes disclosure times for different realized qualities $q$, strictly benefiting from this additional design flexibility. We define the following family of information policies:
\begin{Def}[Cutoff cascade]
An information policy features a \emph{cutoff cascade} if there exists a \emph{cutoff quality} function $\tilde{q}\kh{s}$ that is nondecreasing and ranges in $\fkh{0,\overline{q}}$, such that at each instant $s\ge 0$ if the task is incomplete, the principal tells the agent whether $q\ge \tilde{q}\kh{s}$ or $q<\tilde{q}\kh{s}$. The agent quits when informed of $q<\tilde{q}\kh{s}$ and continues otherwise.  Formally, $\sigma_s\kh{q}=\indic{q < \tilde{q}\kh{s}}$.
\end{Def}
Figure~\ref{fig:cutoff} shows examples of cutoff cascade policies. 
\begin{figure}[h]
\begin{subfigure}[t]{0.49\textwidth}
\includegraphics[width=\textwidth]{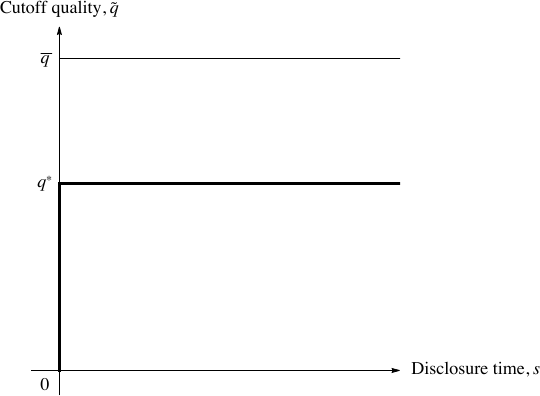}
\caption{Static disclosure with cutoff quality $q^*$.}\label{fig:cutoff1}
\end{subfigure}\hfill
\begin{subfigure}[t]{0.49\textwidth}
\includegraphics[width=\textwidth]{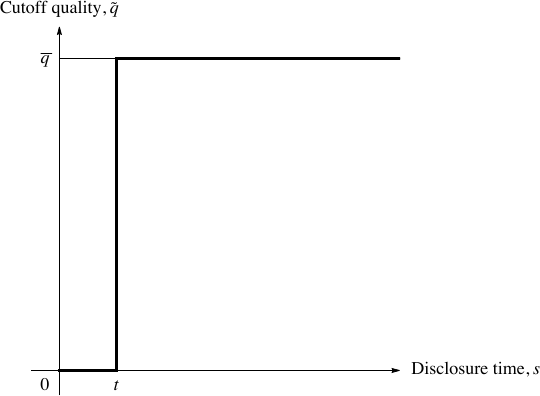}
\caption{One-shot delayed disclosure at time $t$.}	\label{fig:cutoff2}
\end{subfigure}
\begin{subfigure}[t]{0.49\textwidth}
\includegraphics[width=\textwidth]{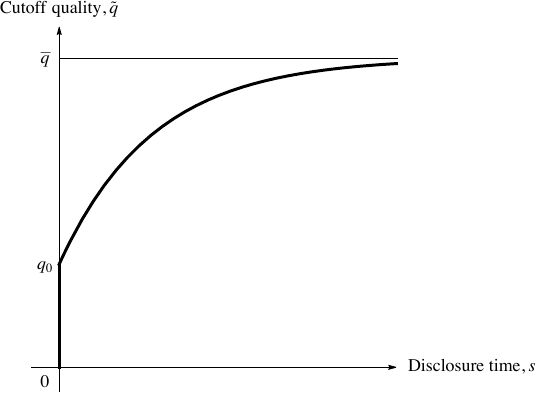}
\caption{Initial disclosure with cutoff quality $q_0$, followed by a series of gradual disclosures.}\label{fig:cutoff3}
\end{subfigure}\hfill
\begin{subfigure}[t]{0.49\textwidth}
\includegraphics[width=\textwidth]{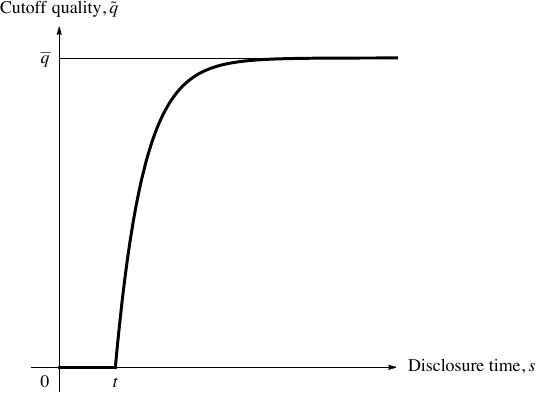}
\caption{Delayed gradual disclosure starting at time $t$.}\label{fig:cutoff4}
\end{subfigure}
\caption{Examples of cutoff cascades.}\label{fig:cutoff}
\end{figure}
$\operatorname{KG}$ corresponds to the constant cutoff quality function $\tilde{q}(s)=q^{*}$, as shown in Figure~\ref{fig:cutoff1}. Figure \ref{fig:cutoff2} plots one-shot delayed disclosure at time $t$, the direct adaptation of the information policy $\operatorname{DD}(t)$ (Definition \ref{def:DD(t)}). Figure \ref{fig:cutoff3} plots \emph{immediate and gradual disclosure} with cutoff quality $q_0$, $\operatorname{IGD}\kh{q_0}$, which combines an immediate disclosure with cutoff quality $q_0$, followed by gradual revelation. Figure \ref{fig:cutoff4} plots  \emph{delayed gradual disclosure} starting at time $t$, $\operatorname{DGD}(t)$, which delays all disclosure until time $t$ before beginning gradual revelation. We show that when $r_p$ is slightly larger than $r$, the optimal policy is $\operatorname{IGD}\kh{q^{**}}$ for some optimal $q^{**}$ (Proposition \ref{prop:gstar}), while when $r_p$ is much larger than $r$, it is $\operatorname{DGD}\kh{\tilde{t}}$ for some optimal $\tilde{t}$ (Proposition \ref{prop:gFDD}).

To understand the logic behind the optimality of cutoff cascade policies, we employ a two-step relaxation approach. We first consider a relaxed version of the principal's optimization problem where we temporarily set aside the interim incentive compatibility constraints. In this relaxed program, the principal can choose an arbitrary joint distribution over $(q,\tau)$ \textendash{} the task quality and the quitting time \textendash{}  subject only to two constraints: (i) the marginal distribution of $q$ matches the prior $F$, and (ii) the agent's individual rationality constraint is satisfied.  Importantly, we ignore whether such a joint distribution can actually be implemented through any feasible information policy. 

We argue that this relaxed program admits an optimal solution that is \textit{positively assortative} \textendash{} higher quality tasks are assigned longer working times. This follows from an optimal assignment principle: since the principal's bonus $b(q)$ is nondecreasing in task quality $q$, each additional moment of work generates higher returns on high-quality tasks than on low-quality tasks. If the principal mistakenly assigned a high-quality task a short working time and a low-quality task a long working time, she could improve her payoff by swapping these assignments without affecting the agent's expected payoff. This positive assortativity naturally leads to the structure of a cutoff cascade: tasks are filtered out in order of increasing quality, with the lowest-quality tasks revealed (and quit) first, while higher-quality tasks continue.

This insight is not merely intuitive \textendash{} it plays a central role in the formal analysis. By establishing positive assortativity in the relaxed program, we reduce the optimization to a tractable family of cutoff cascades. The second step of our approach shows that a positively assortative solution to the relaxed program can indeed be implemented through a cutoff cascade policy, thereby restoring incentive compatibility. This implementability result is formalized in the following Lemma \ref{lem:DGD}. Individual rationality is clearly necessary: the agent must obtain at least zero ex ante expected payoff, since otherwise he can simply ignore all information, quit immediately, and ensure a zero payoff. Lemma \ref{lem:DGD} shows that individual rationality is both necessary and sufficient for the agent to follow cutoff cascade policies.
\begin{Lem}\label{lem:DGD}
Consider any cutoff cascade policy $\bm{\sigma}$. The agent's best response is to follow $\bm{\sigma}$ if and only if $\bm{\sigma}$ is individually rational; that is, following $\bm{\sigma}$  yields a nonnegative ex ante payoff for the agent.
\end{Lem}
The intuition behind Lemma \ref{lem:DGD} relies on the positive assortativity: higher quality leads to later quitting under cutoff cascades. As time passes without being told to quit, the agent's beliefs about task quality improve \textendash{} after all, if the task were low quality, it likely would have been revealed by now. This belief updating makes cutoff cascade policies incentive compatible, as being asked to continue is good news about quality.

Having established that cutoff cascades are without loss of optimality, we now characterize the specific form of the optimal cascade. A key question is whether the principal should begin disclosing immediately or delay. The analysis divides into two cases based on whether a baseline policy \textendash{} \textit{immediate and gradual disclosure} with no initial information, $\operatorname{IGD}(0)$ (formally defined in Subsection~\ref{sec:IGD}) \textendash{} satisfies the agent's individual rationality constraint. When the principal is moderately impatient,  $\operatorname{IGD}(0)$ is not individually rational, and the principal optimally provides an initial disclosure to boost the agent's belief before beginning gradual revelation (Proposition~\ref{prop:gstar}). When the principal is highly impatient, $\operatorname{IGD}(0)$ is individually rational, and the principal can profitably delay before beginning gradual revelation (Proposition~\ref{prop:gFDD}).

\subsection{Immediate and Gradual Disclosure}\label{sec:IGD}
When the principal is moderately impatient ($r_p$ is slightly larger than $r$), she optimally provides an initial disclosure to raise the agent's belief and then begins gradual revelation. To formalize this, we first define \emph{immediate and gradual disclosure} starting from an initial cutoff $q_0$, denoted $\operatorname{IGD}(q_0)$, and then characterize the optimal initial cutoff.

We begin by defining the ratio between the agent's loss and the principal's gain for any task $q\in\fkh{0,\overline{q}}$ that is not individually rational: $u\kh{q}\equiv \abs{v\kh{q}}/{w\kh{q}}$. Since the principal's bonus $b\kh{q}$ is assumed to be nondecreasing, so is $w\kh{q}=\frac{\lambda{b\kh{q}}}{r_p+\lambda}$. Moreover, since $\abs{v\kh{q}}$ decreases from $\abs{v\kh{0}}$ to zero as $q$ rises from zero to $\overline{q}$, the ratio $u\kh{q}$ is strictly decreasing and ranges from the positive value $u\kh{0}$ to zero. Using this ratio, we define a specific cutoff cascade policy as follows.
\begin{Def}[Immediate and gradual disclosure]\label{def:IGD}
$\operatorname{IGD}\kh{q_0}$ denotes the cutoff cascade policy of \emph{immediate and gradual disclosure} with initial cutoff quality $q_0\in\left(0,\overline{q}\right)$, specified by the cutoff quality function
\eqn{\tilde{q}^{\operatorname{IGD}(q_0)}\kh{s}=u^{-1}\kh{u\kh{q_0}e^{-\kh{r_p-r}s}},\label{eqn:qtildeIGD}}
where $u^{-1}\kh{\cdot}$ is the inverse function of $u$.
\end{Def}

Under this policy, the cutoff quality rises continuously from $q_0$ to $\overline{q}$ as time $s$ increases. Solving for the time at which each quality is revealed yields the agent's quitting time:
\eqn{\tau^{\operatorname{IGD}(q_0)}\kh{q}=\Brace{&0,&&q<q_0,\\&\frac{1}{r_p-r}\log\kh{\frac{u\kh{q_0}}{u\kh{q}}},&&q_0\le q<\overline{q},\\&\infty,&&q\ge \overline{q}.}\label{eqn:tauIGD}}
Equation~\eqref{eqn:tauIGD} is obtained by inverting the cutoff quality function \eqref{eqn:qtildeIGD}; these equivalent representations characterize the policy through the timing of disclosure and the schedule of quitting, respectively.

Note that if the principal sets $q_0=q^*$ (the cutoff in $\operatorname{KG}$), $\operatorname{IGD}(q^*)$ reveals strictly more information than $\operatorname{KG}$, thereby granting the agent a strictly positive ex ante payoff. The principal can improve her payoff by lowering the initial cutoff $q_0$ to extract this surplus, but $q_0$ is bounded below by zero.

The analysis therefore divides into two cases, depending on whether $\operatorname{IGD}(0)$ is individually rational. From equation \eqref{eqn:tauIGD}, the quitting time under $\operatorname{IGD}(0)$ for any quality $q\in \kh{0,\overline{q}}$ is $\tau^{\operatorname{IGD}(0)}(q)=\frac{1}{r_{p}-r} \log \left(\frac{u(0)}{u(q)}\right)$. As $r_p$  increases, this quitting time decreases, making the policy more favorable to the agent. When the principal is moderately impatient ($r_p$ is slightly larger than $r$), the quitting times under $\operatorname{IGD}(0)$ are long and the policy  violates individual rationality. When  the principal is highly impatient ($r_p$ is much larger than $r$), the quitting times are shorter and $\operatorname{IGD}(0)$ becomes individually rational.

If $\operatorname{IGD}(0)$ violates individual rationality, the principal chooses an optimal initial cutoff quality $q^{**}\in\kh{0,q^*}$ such that the agent's payoff from $\operatorname{IGD}\kh{q^{**}}$ equals exactly zero. 
\begin{Def}[Optimal immediate and gradual disclosure]
$\operatorname{OIGD}$ denotes the cutoff cascade policy of \emph{optimal  immediate and gradual disclosure}, specified by the cutoff quality function
$$\tilde{q}^{\operatorname{OIGD}}\kh{s}=\tilde{q}^{\operatorname{IGD}\kh{q^{**}}}\kh{s}=u^{-1}\kh{u\kh{q^{**}}e^{-\kh{r_p-r}s}},$$
where $u^{-1}\kh{\cdot}$ is the inverse function of $u$, and $q^{**}\in\kh{0,\overline{q}}$ is chosen such that the agent's ex ante expected payoff equals zero.	
\end{Def}
For such a moderately impatient principal, $\operatorname{OIGD}$ is optimal among all dynamic information policies.
\begin{Prop}\label{prop:gstar}
If $r_p> r$ and	$\operatorname{IGD}(0)$ is not individually rational, then $\operatorname{OIGD}$ is optimal among all dynamic information policies.
\end{Prop}
The relationship between the information policies $\operatorname{KG}$ and $\operatorname{OIGD}$ is shown in Figure \ref{fig:OIGD1}, which corresponds to the case where $r_p$ is slightly larger than $r$, and $\operatorname{OIGD}$ is optimal. As $r_p$ approaches $r$ from above, the initial cutoff quality $q^{**}$ approaches  $q^*$ from below. Moreover, for any $q\in\kh{q^*,\overline{q}}$, the quitting time $\tau^{\operatorname{OIGD}}\kh{q}=\frac{1}{r_p-r}\log\kh{\frac{u\kh{q^{**}}}{u\kh{q}}}$ tends to infinity. Therefore, $\operatorname{OIGD}$ converges to $\operatorname{KG}$, demonstrating continuity in the optimal information policy as the principal's discount rate varies.

\subsection{Delayed Gradual Disclosure}\label{sec:DGD}
We have shown that when $\operatorname{IGD}(0)$ is not individually rational, the principal's optimal policy is $\operatorname{OIGD}$ \textendash{} immediate disclosure followed by gradual revelation. We now turn to the complementary case: when $\operatorname{IGD}(0)$ \emph{is} individually rational. This occurs when the principal is sufficiently impatient ($r_p$ is much larger than $r$), and in this case $\operatorname{IGD}(0)$ yields a strictly positive ex ante payoff for the agent. The principal can extract this surplus by delaying disclosure. To formalize this, we first define \emph{delayed gradual disclosure} starting at time $t$, denoted $\operatorname{DGD}(t)$, and then characterize the optimal delay time.
\begin{Def}[Delayed gradual disclosure]\label{def:DGD}
$\operatorname{DGD}(t)$ denotes the cutoff cascade policy of \emph{delayed gradual disclosure} starting at time $t$, specified by the cutoff quality function
\begin{equation}\tilde{q}^{\operatorname{DGD}(t)}\kh{s}=\Brace{&0,&&s<t,\\&u^{-1}\kh{u\kh{0}e^{-\kh{r_p-r}\kh{s-t}}},&&s\ge t,}\label{eqn:qtildeDGD}\end{equation}
where $u^{-1}\kh{\cdot}$ is the inverse function of $u$.
\end{Def}
Under this policy, the cutoff quality remains zero until time $t$, after which it rises continuously to $\overline{q}$. Solving for the time at which each quality is revealed yields the agent's quitting time:
\eqn{\tau^{\operatorname{DGD}(t)}\kh{q}=\Brace{&t+\frac{1}{r_p-r}\log\kh{\frac{u\kh{0}}{u\kh{q}}},&&q<\overline{q},\\&\infty,&&q\ge \overline{q}.}\label{eqn:tauDGD}}
Equation~\eqref{eqn:tauDGD} is obtained by inverting the cutoff quality function \eqref{eqn:qtildeDGD}; as $q$ rises from $0$ to $\overline{q}$, the agent's quitting time continuously rises from $t$ to infinity.

Let $\tilde{t}$ denote the \emph{maximum delay time}, i.e., the latest time such that $\operatorname{DGD}\kh{\tilde{t}}$ is individually rational for the agent. Since the agent's expected payoff decreases continuously in the delay time, $\tilde{t}$ is unique and finite.
\begin{Def}[Maximally delayed gradual disclosure]
$\operatorname{MDGD}$ denotes the cutoff cascade policy of \emph{maximally delayed gradual disclosure}, specified by the cutoff quality function
$$\tilde{q}^{\operatorname{MDGD}}\kh{s}=\tilde{q}^{\operatorname{DGD}\kh{\tilde{t}}}\kh{s}=\Brace{&0,&&s<\tilde{t},\\&u^{-1}\kh{u\kh{0}e^{-\kh{r_p-r}\kh{s-\tilde{t}}}},&&s\ge \tilde{t},}$$
where $u^{-1}\kh{\cdot}$ is the inverse function of $u$, and $\tilde{t}\ge 0$ is chosen such that the agent's ex ante expected payoff equals zero.	
\end{Def}
For a highly impatient principal, $\operatorname{MDGD}$ is optimal among all dynamic information policies.
\begin{Prop}\label{prop:gFDD}
If $r_p> r$ and	$\operatorname{IGD}(0)$ is individually rational, then $\operatorname{MDGD}$ is optimal among all dynamic information policies.\end{Prop}
\begin{figure}[!htbp]\centering
\begin{subfigure}[t]{0.49\textwidth}
\includegraphics[width=\textwidth]{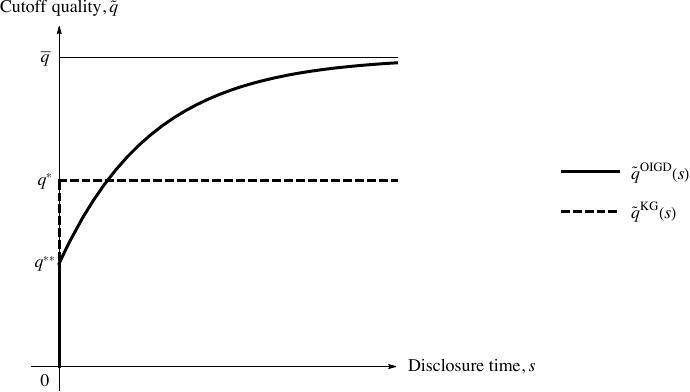}
\caption{Optimal immediate and gradual disclosure ($\operatorname{OIGD}$).}\label{fig:OIGD1}
\end{subfigure}\hfill
\begin{subfigure}[t]{0.49\textwidth}
\includegraphics[width=\textwidth]{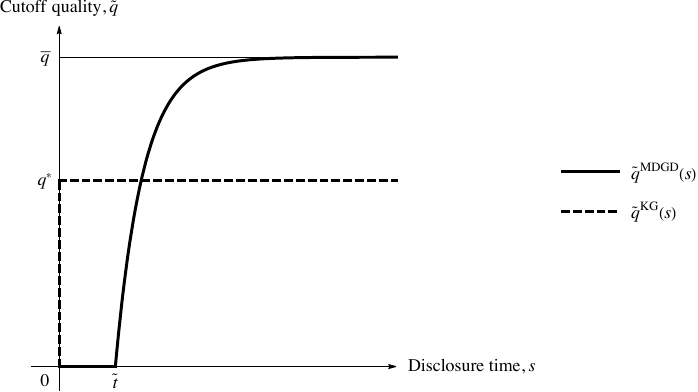}
\caption{Maximally delayed gradual disclosure ($\operatorname{MDGD}$).}	\label{fig:MDGD1}
\end{subfigure}
\caption{Relationship among $\operatorname{KG}$, $\operatorname{OIGD}$ and $\operatorname{MDGD}$.}\label{fig:MDGD}
\end{figure}
The relationship between $\operatorname{KG}$ and $\operatorname{MDGD}$ is shown in Figure \ref{fig:MDGD1}, corresponding to the case where  $r_p$ is much larger than $r$ and $\operatorname{MDGD}$ is optimal. As $r_p$ increases beyond the range where  $\operatorname{OIGD}$ is optimal, the principal shifts from lowering the initial cutoff quality to introducing a delay before disclosure begins. The notable difference from the binary case (Section~\ref{sec:bin}) is that even when delay is optimal, the policy involves gradual rather than one-shot disclosure. The principal customizes disclosure times for different quality realizations, extracting more value than any one-shot delayed policy.

\section{Binary-Nonstationary Case}\label{sec:Poi}
We now move to nonstationary environments by allowing the task completion rate to vary with task quality, while maintaining the binary quality assumption from Section \ref{sec:bin}. In this setting, the mere passage of time conveys information about the task's quality.

The analysis divides into two cases, depending on whether the completion rate of the high-quality task is larger or smaller than that of the low-quality task. We focus on the \emph{pessimistic case} ($\lambda_H>\lambda_L$), where the high-quality task has a shorter expected duration than the low-quality task. In this case,   ``no news is bad news'' \textendash{} as time passes without completion, the agent becomes increasingly pessimistic about the task's quality. The analysis reveals a new channel that makes dynamic disclosures valuable: the agent becomes  \emph{pessimistic} over time.\footnote{By contrast, in the \emph{optimistic case} ($\lambda_H< \lambda_L$) where ``no news is good news", all the results from Section \ref{sec:bin}  carry over directly: the principal benefits from dynamic disclosure if and only if she is less patient than the agent ($r_p>r$), and $\operatorname{MDD}$ is optimal for an impatient principal. In fact, if the agent starts working, he will never find it optimal to quit early, so the optimistic case reduces to an essentially stationary problem.}

In the pessimistic case, the principal benefits from dynamic disclosure \emph{regardless of the relationship between $r_p$ and $r$}. However, Proposition \ref{prop:optlearn} shows that the structure of the optimal information policy still depends on the patience levels. While $\operatorname{MDD}$ remains optimal for an impatient principal ($r_p>r$),   a patient principal ($r_p\le r$) optimally employs \emph{Poisson disclosure}, where full disclosure arrives at the calibrated Poisson rate $\kh{\lambda_H-\lambda_L}$ when the task has low quality. This calibrated rate ensures that the absence of disclosure neither raises nor lowers the agent's beliefs \textendash{} effectively restoring stationarity.

The section proceeds as follows. Subsection~\ref{sec:agent} introduces the critical belief threshold $\hat{\mu}$ by analyzing the agent's behavior in the absence of information provision; this threshold shapes the structure of optimal policy throughout the section. Subsection~\ref{sec:free} establishes the Free-Riding Lemma (Lemma~\ref{lem:free}), showing that for high priors ($\mu > \hat{\mu}$), the principal optimally remains silent until the belief decays to $\hat{\mu}$. Subsection~\ref{sec:poisson} introduces Poisson disclosure, which stabilizes the agent's belief by offsetting the natural pessimistic drift. Subsection~\ref{sec:optPoi} states the main optimality result: the form of optimal policy still depends crucially on relative patience.

\subsection{Agent Behavior and the Critical Belief Threshold}\label{sec:agent}
In the pessimistic case, when a task remains unfinished for an extended period, the agent infers that it is likely of low quality. This creates a fundamental challenge: even an agent who begins with high beliefs about the task may eventually lose faith and quit before completion, creating a dynamic problem that the principal must solve.

We first consider the decision problem of an \emph{uninformed} agent \textendash{} one who receives no disclosure from the principal. This benchmark characterizes behavior in the absence of any information policy and will be essential for understanding the value of disclosure. Bayes' rule implies that if the task is not completed at time $t\ge 0$, the agent's belief $\mu_{t}$ is given by
\begin{align}
\mu_{t} \equiv \mathbb{P}(q=H \mid x>t) &= \frac{\mathbb{P}(x>t \mid q=H)\mathbb{P}(q=H)}{\mathbb{P}(x>t \mid q=H)\mathbb{P}(q=H)+\mathbb{P}(x>t \mid q=L)\mathbb{P}(q=L)} \notag \\
&= \frac{\mu}{\mu+e^{(\lambda_H-\lambda_L)t}(1-\mu)}, \label{eqn:mu_t}
\end{align}
which decreases over time since $\lambda_H>\lambda_L$. Let $\hat{\mu}$ denote threshold belief where
\eqn{\hat{\mu}\left(r+\lambda_{H}\right) v(H)+\left(1-\hat{\mu}\right)\left(r+\lambda_{L}\right) v(L)=0.\label{eqn:muhat}}
The threshold $\hat{\mu}$ plays a central role in all subsequent results, determining both the agent's behavior and the structure of optimal policy throughout this section. If the agent's prior is below $\hat{\mu}$, he quits immediately. Otherwise, he begins working, but with a contingency plan: if the task remains incomplete for too long, his belief will drop to $\hat{\mu}$, at which point it becomes optimal to quit.

Let $\overline{t}(\mu)$ denote the time when the agent's belief falls to $\hat{\mu}$ starting from a prior $\mu>\hat{\mu}$. Setting $\mu_t=\hat{\mu}$ in equation~\eqref{eqn:mu_t} and solving yields:
\eqn{\frac{\mu}{\mu+e^{\kh{\lambda_H-\lambda_L}\overline{t}(\mu)}\kh{1-\mu}}=\hat{\mu}\quad\iff\quad \overline{t}(\mu)=\frac{1}{\lambda_{H}-\lambda_{L}}\left(\log \frac{\mu}{1-\mu}-\log \frac{\hat{\mu}}{1-\hat{\mu}}\right)\label{eqn:tbar}.}
Lemma \ref{lem:tbar} characterizes the agent's behavior in the absence of information provision.
\begin{Lem}\label{lem:tbar}
Suppose that the principal does not provide any information to the agent.
\begin{enumerate}
\item If $\mu \le\hat{\mu}$, then the agent quits immediately.
\item If $\mu >  \hat{\mu}$, then the agent starts working but plans to quit at time $\overline{t}\kh{\mu}$ if the task remains incomplete.	
\end{enumerate}
\end{Lem}
Beyond characterizing when the uninformed agent quits, $\hat{\mu}$ determines the structure of optimal disclosure \textendash{} in particular, whether the optimal policy involves delay. As we formalize in the Free-Riding Lemma (Lemma~\ref{lem:free}), when the agent's prior exceeds $\hat{\mu}$, the principal optimally remains silent until the belief decays to $\hat{\mu}$, free-riding on the agent's willingness to continue without information provision.

The threshold $\hat{\mu}$ also reveals the limitations of static disclosure, demonstrating the need for dynamic information provision. Under a static policy, the agent may still quit the high-quality task before completion. For example, Figure \ref{fig:WKGl} illustrates the optimal static information policy (which we still call $\operatorname{KG}$).\begin{figure}[!htbp]\centering
\includegraphics[height=0.3\textheight]{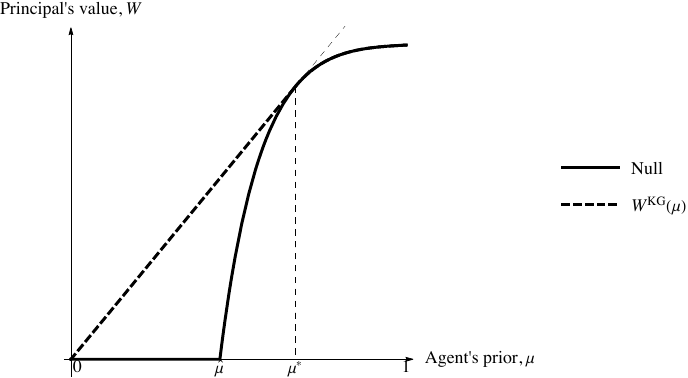}
\caption{Optimal static information policy $\operatorname{KG}$ corresponds to a concavification.}\label{fig:WKGl}
\end{figure}
The solid curve represents the principal's value without persuasion. For $\mu\le \hat{\mu}$, the agent quits immediately, and the principal's value is zero. As $\mu$ rises from $\hat{\mu}$ to one, the agent works increasingly longer, so the principal's value also rises. The dashed line segment, formed by the concavification of the solid curve, represents the principal's gain from persuasion. In particular, $\operatorname{KG}$  raises the agent's belief to some $\mu^*>\hat{\mu}$ to induce working, but the agent still plans to quit at time $\overline{t}\kh{\mu^*}$ if the task has not been completed by then. Under any static information policy, the only way to guarantee completion of the high-quality task is initial full disclosure (to push the belief to one), which is too costly.

Dynamic disclosure overcomes this limitation. Continuing the above example, suppose that in addition to the initial disclosure that motivates the agent to start working, the principal also commits to fully disclose the task quality at time $\overline{t}\kh{\mu^*}$ (the agent's planned quitting time) if the task is still incomplete. If the task has high quality, the agent will learn this at $\overline{t}\kh{\mu^*}$ and thus continue working until eventual completion. If the task has low quality, the principal's disclosure at $\overline{t}\kh{\mu^*}$ simply causes the agent to quit (which he was going to do anyway). Crucially, the agent knows this disclosure plan from the start, so he is willing to work until time $\overline{t}\kh{\mu^*}$ since the dynamic plan only strengthens his incentives to continue compared to the static policy. For the same reason, following this information policy yields a strictly positive ex ante payoff for the agent, giving the principal room for improvement.

\subsection{The Free-Riding Lemma}\label{sec:free}
We now analyze the optimal structure of dynamic information disclosure. The critical step in this nonstationary setting is identifying the binding incentive constraints within the continuum of constraints the agent faces over time. A key insight emerges: if the agent has a high prior ($\mu >  \hat{\mu}$), the incentive constraint at time $\overline{t}\kh{\mu}$ implies all earlier constraints, regardless of the information policy. Therefore, it is optimal for the principal to provide no information and ``free-ride'' until the belief falls to $\hat{\mu}$. This observation is formalized in the Lemma \ref{lem:free} below.
\begin{Lem}[Free-Riding]\label{lem:free}
If $\mu >  \hat{\mu}$, any optimal information policy provides no information at all $t < \overline{t}(\mu)$. That is, the agent never quits and is never told to quit before the belief naturally decays to $\hat{\mu}$.
\end{Lem}
In light of this Free-Riding Lemma, once we obtain the optimal information policy for prior $\mu=\hat{\mu}$, we can solve the problem for any $\mu>\hat{\mu}$ by translating the time axis: relabel time  $\overline{t}\kh{\mu}$ as time zero and adapt the optimal policy accordingly. We therefore focus on characterizing the optimal policy at or below the critical prior $\hat{\mu}$.

\subsection{Poisson Disclosure}\label{sec:poisson}
To sustain the agent's effort at the critical belief $\hat{\mu}$, the principal must prevent the belief from drifting downward. Consider the following policy: instead of setting a deterministic disclosure time, the principal discloses the task quality at Poisson rate $\kh{\lambda_H-\lambda_L}$ when the task has low quality. We first show that this  \emph{Poisson disclosure} stabilizes the agent's belief over time.

Suppose that the agent neither receives disclosure nor completes the task during a small time interval $\Delta{t}$. This event occurs with probability $\kh{1-\lambda_H\Delta{t}}$ if the task quality is high, and with probability $$\kh{1-\lambda_L\Delta{t}}\kh{1-\kh{\lambda_H-\lambda_L}\Delta{t}}=1-\lambda_H\Delta{t}$$ if the task quality is low. Since the event occurs is equally likely under both task qualities, Bayes' rule implies that the agent does not update his belief. In this way, ``no news is no news'', and the stationarity of the problem is restored. Thus, Poisson disclosure plays the same role as keeping silence in the stationary case analyzed in Section \ref{sec:bin}.

If $\mu < \hat{\mu}$, Poisson disclosure at rate $(\lambda_H - \lambda_L)$ alone cannot induce the agent to work. The principal addresses this by combining Poisson disclosure with an initial announcement that raises the agent's belief to exactly $\hat{\mu}$ \textendash{} the threshold at which he is just willing to start. This initial announcement follows the same logic as $\operatorname{KG}$: it pools high-quality with low-quality to boost the agent's belief.  Once the agent begins working, his belief remains constant at $\hat{\mu}$ as long as no disclosure arrives, so he continues. He quits only when informed that the task has low quality. We refer to this information policy as \emph{immediate  and Poisson disclosure}, denoted by $\operatorname{IPD}$. 
\begin{Def}[Immediate  and Poisson disclosure]\label{def:IPD}
For $\mu\le \hat{\mu}$, $\operatorname{IPD}$ denotes the information policy of \emph{immediate and Poisson disclosure} where 
\begin{enumerate}
    \item At time zero, the agent is told to continue ($\sigma_0=0$) with probability one when $q=H$, and also with positive probability when $q=L$, such that $\mathbb{P}(q=H \mid \sigma_0=0)=\hat{\mu}$.
    \item For $s>t$, conditional on $\sigma_t=0$, the recommendation switches to quit ($\sigma_s=1$) according to a Poisson arrival with rate $(\lambda_H-\lambda_L)$ when $q=L$. Formally, $\mathbb{P}(\sigma_s=1 \mid \sigma_t=0) = \kh{1-e^{-(\lambda_H-\lambda_L)(s-t)}}\cdot\indic{q=L}$.
\end{enumerate}
\end{Def}
Under $\operatorname{IPD}$, once the agent starts working, the agent's quitting time follows an exponential distribution with rate $(\lambda_H-\lambda_L)$ when the task has low quality, whereas he never quits when the task has high quality. Part~\eqref{1a} of Proposition \ref{prop:optlearn} below shows that when the principal is more patient than the agent ($r_p\le r$), $\operatorname{IPD}$ is optimal among all dynamic information policies for low priors ($\mu\le \hat{\mu}$).

For high priors ($\mu > \hat{\mu}$), the optimal policy for a patient principal is constructed by combining this result with the Free-Riding Lemma (Lemma \ref{lem:free}): the principal provides no information until time $\overline{t}(\mu)$ and switches to IPD thereafter. This constitutes \emph{delayed Poisson disclosure}, denoted by $\operatorname{DPD}$, where Poisson disclosure begins only after $\overline{t}\kh{\mu}$.
\begin{Def}[Delayed Poisson disclosure]\label{def:DPD}
For $\mu>\hat{\mu}$, $\operatorname{DPD}$ denotes the information policy of \emph{delayed Poisson disclosure}. The principal provides no information until $\overline{t}\kh{\mu}$, and then discloses at Poisson rate $\kh{\lambda_H-\lambda_L}$ if $q=L$. Formally, $\sigma_s=0$ for all $s\le \overline{t}\kh{\mu}$, and for any $s > t \ge \overline{t}\kh{\mu}$, $ \mathbb{P}\kh{\sigma_s=1 \mid \sigma_t=0} = \kh{1 - e^{-\kh{\lambda_H-\lambda_L}\kh{s-t}}} \cdot \indic{q=L}. $
\end{Def}

\subsection{Optimal Information Policy}\label{sec:optPoi}
We now state the main result characterizing the optimal policy, Proposition~\ref{prop:optlearn}. While Poisson disclosure achieves belief stationarity, the principal is not restricted to stochastic disclosure; deterministic policies (such as $\operatorname{MDD}$ from Section \ref{sec:bin}) remain available. Part~\ref{2} of Proposition~\ref{prop:optlearn} establishes that for an impatient principal ($r_p>r$), $\operatorname{MDD}$ is optimal, paralleling Proposition \ref{prop:opt}. This highlights that while dynamic disclosure is always valuable in a pessimistic environment, the \textit{form} of optimal disclosure still depends crucially on relative patience.

For any $\mu\le\hat{\mu}$, the maximum delay time $\tilde{t}(\mu)$ remains defined by equation \eqref{eqn:ttilde}. For $\mu>\hat{\mu}$, the Free-Riding Lemma (Lemma \ref{lem:free}) implies that persuasion only takes effect from time $\overline{t}\kh{\mu}$ onward. Therefore, the maximum delay time for a high prior is the free-riding duration plus the standard maximum delay at $\hat{\mu}$: 
\eqn{\tilde{t}\kh{\mu}=\overline{t}\kh{\mu}+\tilde{t}\kh{\hat{\mu}}\quad\text{for }\mu>\hat{\mu}.\label{eqn:ttilde1}}
With $\tilde{t}(\mu)$ now defined for all $\mu \in (0,1)$, we extend the Definition~\ref{def:FDD} of MDD to this nonstationary setting:
\begin{Defprime}{def:FDD}[Maximally delayed disclosure]\label{def:MDDl}
$\operatorname{MDD}$ denotes the information policy of \emph{maximally delayed disclosure}, where   for any $\mu\in\kh{0,1}$, the principal provides no information prior to the time $\tilde{t}(\mu)$, and fully discloses the task quality at $\tilde{t}(\mu)$. Formally, $\sigma_s = 0$ for all $s < \tilde{t}(\mu)$, and for all $s \ge \tilde{t}(\mu)$, $\sigma_s = \indic{q=L}$.
\end{Defprime}
Proposition \ref{prop:optlearn} establishes the optimality of these information policies.\footnote{We thank an anonymous referee for pointing out the connection between Proposition \ref{prop:optlearn} (part \ref{1}) and Proposition 8 in Section VI of \cite{ES20}, as both results feature exponentially distributed time lotteries with calibrated parameter. However, neither result directly nests the other. Section VI of \cite{ES20} considers the problem where the principal can both \textit{design} the distribution of task durations and \textit{observe} its realization before choosing an optimal information policy. In our model, the principal can neither choose nor ex ante observe the task duration. Nonetheless, in the special case of $\lambda_L=0$ and $\mu=\hat{\mu}$, one can deduce from Lemma 3 and Proposition 8 in \cite{ES20} that if $r_p\le r$, then the time lottery induced by Poisson disclosure at rate $\lambda_H$ (which equals $\lambda_H-\lambda_L$ when $\lambda_L=0$) is optimal in our model.} When $r_p\le r$, either $\operatorname{IPD}$  or $\operatorname{DPD}$ is optimal  depending on the prior. When $r_p>r$, $\operatorname{MDD}$ is optimal.
\begin{Prop}\label{prop:optlearn}In the pessimistic case with $\lambda_H>\lambda_L$,
\begin{enumerate}
\item\label{1} If $r_p\le r$, then the optimal policy involves Poisson disclosure.
\begin{enumerate}
\item\label{1a} If $\mu\le  \hat{\mu}$, then $\operatorname{IPD}$  is optimal among all dynamic information policies.
\item\label{1b} If $\mu>  \hat{\mu}$, then $\operatorname{DPD}$ is optimal among all dynamic information policies.
\end{enumerate}
\item\label{2} If $r_p> r$, then $\operatorname{MDD}$ is optimal among all dynamic information policies.
\end{enumerate}	
\end{Prop}
The intuition mirrors the stationary case (Subsection~\ref{subsec:opt}): the optimal policy is determined by the parties' risk attitudes towards time lotteries. The key difference here is the natural pessimistic drift in the agent's belief absent disclosure. Any information policy that ensures completion of the high-quality task must reveal low quality at a rate of at least $(\lambda_H-\lambda_L)$ to offset this downward drift (i.e., to make ``no news no news''). The principal can choose to disclose at a faster rate, with full disclosure at a deterministic time representing the limiting case of an infinite revelation rate.

The structure of binding constraints differs between the two cases. Under Poisson disclosure ($\operatorname{IPD}$ or $\operatorname{DPD}$), multiple dynamic obedience constraints bind simultaneously: throughout the Poisson phase, the agent's belief remains constant at $\hat{\mu}$, so his continuation value equals zero at every instant. By contrast, under $\operatorname{MDD}$, only a single constraint binds: the individual rationality constraint when $\mu \le \hat{\mu}$, or the time-$\overline{t}(\mu)$ constraint when $\mu > \hat{\mu}$ (as implied by the Free-Riding Lemma). In either case, the remaining constraints hold with slack. This difference shapes the proof strategies: establishing optimality of Poisson disclosure requires working with a continuum of incentive constraints and showing that all of them bind at optimality, whereas for $\operatorname{MDD}$, it suffices to identify a single binding constraint and verify that all others are slack.

Figure \ref{fig:WPD} depicts the relationship among $\operatorname{IPD}$/$\operatorname{DPD}$, $\operatorname{MDD}$ and $\operatorname{KG}$ when $r_p<r$.
\begin{figure}[!htbp]
\begin{subfigure}[t]{0.49\textwidth}
\includegraphics[width=\textwidth]{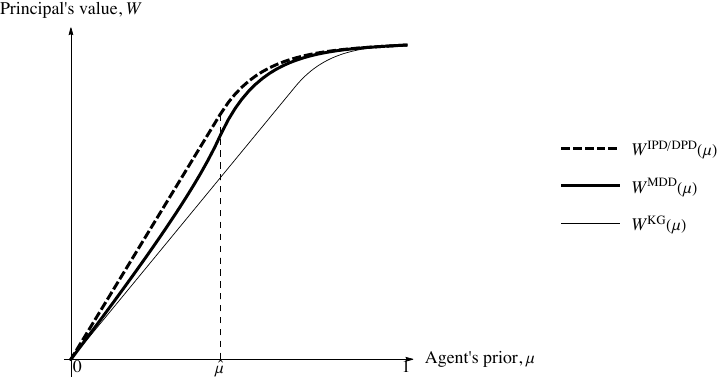}
\caption{$r_p<r$.}\label{fig:WPD}
\end{subfigure}\hfill
\begin{subfigure}[t]{0.49\textwidth}
\includegraphics[width=\textwidth]{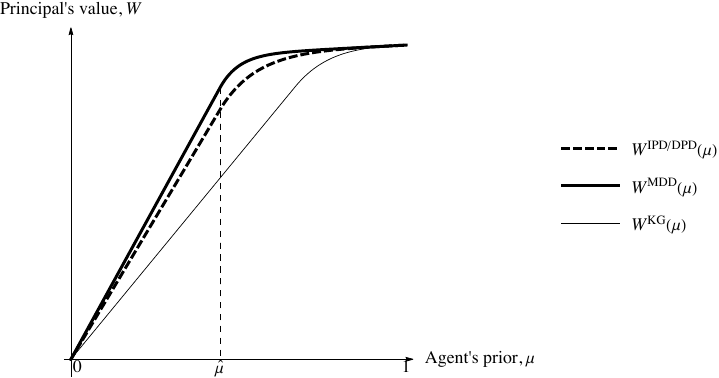}
\caption{$r_p>r$.}\label{fig:WPD1}
\end{subfigure}
\caption{The principal's payoff comparison among $\operatorname{IPD}$/$\operatorname{DPD}$, $\operatorname{MDD}$ and $\operatorname{KG}$.}
\end{figure} 
Both the Poisson policy  and $\operatorname{MDD}$ strictly outperform $\operatorname{KG}$ for all $\mu \in (0,1)$. When $r_p<r$, the agent's risk attitude  dominates, so the optimal lottery has the highest risk. Among policies ensuring incentive compatibility, the Poisson policy  induces the highest-risk time lottery in the sense of second-order stochastic dominance. Conversely, when $r_p>r$, the principal's risk attitude dominates, making the optimal lottery degenerate at a fixed time $\tilde{t}\kh{\mu}$, induced by $\operatorname{MDD}$. As Figure \ref{fig:WPD1} shows, both dynamic policies still improve upon $\operatorname{KG}$, but Poisson policy becomes suboptimal here  due to the excessive risk it imposes.

\section{Conclusion}\label{sec:concln}
In this paper, we study the optimal mechanism to motivate effort in a dynamic moral hazard model without transfers. We develop a unified framework and characterize the optimal policy in closed form, identifying two precise conditions under which dynamic disclosure is strictly valuable to the principal: one is that the principal is impatient compared with the agent, and the other is that the agent would become pessimistic over time absent any disclosure.

Our results reveal distinct optimal policy structures across different environments. In binary-stationary settings, an impatient principal optimally employs maximally delayed disclosure \textendash{} withholding all information until a specific threshold before complete revelation \textendash{}  while a patient principal prefers static disclosure. With general quality distributions, the optimal policy for an impatient principal becomes a cutoff cascade, where the principal gradually reveals whether the reward exceeds a time-varying threshold. In pessimistic environments, maximally delayed disclosure remains optimal for an impatient principal, while a patient principal chooses Poisson disclosure \textendash{}  releasing information at a calibrated rate that exactly offsets the downward drift in the agent's belief. Crucially, our unified framework synthesizes these policy forms by identifying which constraints bind: static individual rationality in stationary environments versus dynamic obedience in pessimistic ones.

Our analysis also emphasizes how players' risk attitudes toward time lotteries shape optimal information policies. Under exponential discounting, these risk attitudes reduce to simple comparisons of discount rates. With other time preferences, however, this relationship becomes considerably richer. This connection between dynamic information design and time preferences suggests a promising direction for future research, particularly in environments where standard exponential discounting may not hold.

Beyond time preferences, Table~\ref{tab:summary} points to another natural direction for future research: the general nonstationary case, where the quality is continuously distributed and the task completion rate varies with quality. This setting combines the challenges of continuous quality (requiring gradual disclosure) with time-varying beliefs (requiring dynamic obedience constraints). Characterizing closed-form solutions in such environments remains open.

\bibliographystyle{myref}
\bibliography{ME.bib}

\renewcommand{\theequation}{\thesection.\arabic{equation}}
\setcounter{equation}{0}
\renewcommand{\theLem}{\thesection.\arabic{Lem}}
\setcounter{Lem}{0}
\renewcommand{\theDef}{\thesection.\arabic{Def}}
\setcounter{Def}{0}
\renewcommand{\theProp}{\thesection.\arabic{Prop}}
\setcounter{Prop}{0}

\appendix
\newpage
\setstretch{1}
\section{Proofs of Results in the Main Text}\label{app:proof}
For any information policy $\bm\sigma$, let $\tau$ denote the induced random quitting time, and let $\tau\kh{q}$ denote the conditional distribution of $\tau$ given task quality $q$.
\subsection{Proofs for Section \ref{sec:bin}}
\begin{proof}[Proof of Lemma~\ref{lem:1}]
We first express the principal's payoff from KG as a function of the agent's prior $\mu$. 
For $\mu<\overline{\mu}$, the KG policy induces a posterior of exactly $\overline{\mu}$ with maximum probability, and a posterior of 0 otherwise. By the standard concavification result, the probability of inducing posterior $\overline{\mu}$ equals $\mu /\overline{\mu}$. Thus,
\eqn{W^{\operatorname{KG}}\kh{\mu} = \kh{\mu /\overline{\mu}}\cdot \left( \overline{\mu} w(H) + (1-\overline{\mu}) w(L) \right)=\mu w\kh{H}+\mu\cdot\frac{1-\overline{\mu}}{\overline{\mu}}w\kh{L}\quad\text{if }\mu<\overline{\mu}.\label{eqn:WKG}}
For $\mu\ge\overline{\mu} $, the agent always starts out completing the task, yielding
$$W^{\operatorname{KG}}\kh{\mu}=\mu w\kh{H}+\kh{1-\mu}w\kh{L}\quad\text{if } \mu\ge\overline{\mu}.$$

Now we consider the principal's payoff from $\operatorname{DD}\kh{t}$. If the agent's prior is too low (below $\tilde{\mu}\kh{t}$), he quits immediately and the principal's payoff is zero, strictly worse than KG. If the prior exceeds $\tilde{\mu}\kh{t}$, $\operatorname{DD}(t)$ induces $\tau\kh{H}=\infty$ and $\tau\kh{L}=t$. Thus, the principal's payoff in this range is 
\eqn{W^{\operatorname{DD}\kh{t}}\kh{\mu}=\mu w\kh{H}+\kh{1-\mu }w\kh{t,L}\quad\text{if }\mu\ge \tilde{\mu}\kh{t},\label{eqn:WDDt}}
where $\tilde{\mu}\kh{t}<\overline{\mu}$. Since $w\kh{t,L}<w\kh{L}$ for all finite $t$, we have $W^{\operatorname{DD}(t)}(\mu)<W^{\operatorname{KG}}\kh{\mu}$ for $\mu\ge\overline{\mu}$ . The only range where $\operatorname{DD}(t)$ can outperform KG is $\left[\tilde {\mu}\kh{t},\overline{\mu}\right)$. Comparing equation \eqref{eqn:WDDt} with equation \eqref{eqn:WKG}, for any $\mu\in \left[\tilde {\mu}\kh{t},\overline{\mu}\right)$, we have
\eqns{W^{\text {DD}(t)}\kh{\mu}>W^{\text{KG}}\kh{\mu}\quad&\iff\quad\kh{1-\mu }w\kh{t,L}>\mu\cdot\frac{1-\overline{\mu}}{\overline{\mu}}w\kh{L}\\
&\iff\quad\kh{1-\mu}\kh{1-e^{-\kh{r_p+\lambda}t}}>\mu\cdot\frac{1-\overline{\mu}}{\overline{\mu}},}
where the last equivalence follows from the principal's payoff function, equation \eqref{eqn:w}. The last inequality imposes an upper bound on $\mu$, so we need the strict inequality to hold at $\mu=\tilde{\mu}(t)$, i.e., 
\eqn{\kh{1-\tilde{\mu}(t)}\kh{1-e^{-\kh{r_p+\lambda}t}}>\tilde{\mu}(t)\cdot\frac{1-\overline{\mu}}{\overline{\mu}}.\label{eqn:A6}}
Comparing equations \eqref{eqn:mubar} and \eqref{eqn:mutilde},  which define the threshold priors $\overline{\mu}$ and $\tilde{\mu}(t)$, we get
\eqns{\kh{1-\tilde{\mu}(t)}\kh{1-e^{-\kh{r+\lambda}t}}=\tilde{\mu}(t)\frac{v\kh{H}}{\fkh{-v\kh{L}}}=\tilde{\mu}(t)\cdot\frac{1-\overline{\mu}}{\overline{\mu}}.}
Therefore, strict inequality \eqref{eqn:A6} holds if and only if
\eqns{1-e^{-\kh{r_p+\lambda}t}>1-e^{-\kh{r+\lambda}t}\quad\iff\quad r_p>r.}
This completes the proof.
\end{proof}

\begin{proof}[Proof of Proposition~\ref{prop:opt}]
For $\mu\ge \overline{\mu}$, both KG and MDD provide no information, and the agent completes the task with probability one without any persuasion.  This is clearly optimal. Henceforth, we assume $\mu<\overline{\mu}$.  

We introduce a change of variables to linearize the principal's objective function.\footnote{We thank an anonymous referee for suggesting this change of variables that streamlines the proof.} Since the principal's value from quitting time $\tau$ is proportional to $1 - e^{-(r_p + \lambda)\tau}$, we define the transformed variable $\zeta$ as $\zeta\equiv 1 - e^{-(r_p + \lambda)\tau}\in[0,1]$. In this transformed space,  $\zeta = 0$ corresponds to immediate quitting  ($\tau = 0$), while $\zeta=1$ corresponds to never quitting ($\tau=\infty$). The principal's value then becomes $\mathbb{E}_{\bm\sigma}\fkh{w\kh{\tau,q}}=\mathbb{    E}_{\bm\sigma}\fkh{w\kh{q}\cdot \zeta}$.

Next we express the agent's value in terms of $\zeta$. The agent's value from quitting time $\tau$ is proportional to $1 - e^{-(r + \lambda)\tau}=1-(1-\zeta)^{\frac{r+\lambda}{r_p+\lambda}}$. Define $\rho\equiv \frac{r_p+\lambda}{r+\lambda}$ as the ratio of the principal's effective discount rate to the agent's, and let $\phi(z)\equiv 1-(1-z)^{1/\rho}$. The agent's value can then be written as  $\mathbb{E}_{\bm\sigma}\fkh{v\kh{\tau,q}}=\mathbb{E}_{\bm\sigma}\fkh{v\kh{q}\cdot \phi(\zeta)}$. Note that $\phi'(z) = \frac{1}{\rho} (1-z)^{1/\rho -1}>0$, so $\phi(z)$ is strictly increasing, and $\phi''(z) = \frac{\rho-1}{\rho^2} (1-z)^{1/\rho-2}$, whose sign depends entirely on $\rho - 1$, i.e., whether the principal is more patient than the agent ($r_p>r$).

Because the agent never needs encouragement to continue when the task quality is high, the non-trivial design problem concerns only the low-quality task. We therefore set $\zeta =1$ when the task quality is high and focus entirely on determining the optimal distribution of $\zeta$ for the low-quality task. Let $G$ denote this distribution of $\zeta$ conditional on low quality. The principal's expected payoff is then
$$W(G)\equiv\mathbb{E}_{\bm\sigma}\fkh{w\kh{q}\cdot \zeta}= \mu w(H) + (1-\mu) w(L) \mathbb{E}_G [\zeta].$$
Since the first term is constant and $w(L) > 0$,  the principal's problem reduces to maximizing $\mathbb{E}_G [\zeta]$. 

The agent's expected payoff is 
$$V(G)\equiv\mathbb{E}_{\bm\sigma}\fkh{v\kh{q}\cdot \phi(\zeta)}= \mu v(H) + (1-\mu) v(L) \mathbb{E}_G\fkh{\phi(\zeta)}.$$
The principal's optimization requires individual rationality $V(G)\ge 0$. Since $v(L) < 0$, we can rewrite this as
\eqns{\mathbb{E}_G\fkh{\phi(\zeta)}\le\frac{\mu v\kh{H}}{\kh{1-\mu}\fkh{-v\kh{L}}}\equiv C_1.}

To find the optimal information policy, we adopt the following approach: we first solve a relaxed problem that temporarily omits all interim incentive compatibility constraints,
\eqn{\begin{split}\max\limits_{G\in\Delta\kh{\fkh{0,1}}}\quad& \mathbb{E}_G [\zeta],\\
\text{s.t.}\quad\quad & \mathbb{E}_G\fkh{\phi(\zeta)}\le C_1,\\
\end{split}\label{eqn:relaxprog}}
where the principal chooses the optimal distribution $G$ subject only to individual rationality. We then verify that the solution to \eqref{eqn:relaxprog} can be implemented by an information policy, thereby establishing its optimality.

If $r_p<r$, then $\rho< 1$ and $\phi(z)$ is strictly concave. By Jensen's inequality, any interior point $\zeta \in (0, 1)$ yields a higher $\phi(\zeta)$ than the equivalent mean-preserving combination of $0$ and $1$, so replacing interior points with such randomizations relaxes the constraint in \eqref{eqn:relaxprog}. Therefore, the optimal distribution  $G$ must be bang-bang.  This is precisely the structure of KG: $\zeta\in\{0, 1\}$ with probabilities that exactly bind individual rationality, implying that KG is optimal among all dynamic information policies.

If $r_p>r$, then $\rho>1$ and $\phi(z)$ is strictly convex. By Jensen's inequality, $\mathbb{E}_G [\phi(\zeta)] \ge \phi(\mathbb{E}_G [\zeta])$, and equality holds if and only if $G$ is degenerate  at a specific point. It follows that $\mathbb{E}_G [\zeta]\le \phi^{-1}\kh{\mathbb{E}_G [\phi(\zeta)] }\le \phi^{-1}\kh{C_1}$, with equality if and only if $G$  is degenerate at $z=\phi^{-1}\kh{C_1}$.  This is precisely the structure of MDD: a deterministic value $\zeta=z$ that exactly binds individual rationality, implying that MDD is optimal among all dynamic information policies.

If  $r_p=r$, then $\rho=1$ and $\phi(z)=z$, making the objective and constraint in \eqref{eqn:relaxprog} identical; thus any policy sufficient to motivate the agent while exactly binding individual rationality is optimal, including both KG and MDD.
\end{proof}

\subsection{Proofs for Section \ref{sec:gen}}\label{app:proofgen}
As discussed in the main text, because we assume that the task quality follows a full-support and atomless distribution, it is without loss to assume that each $\tau\kh{q}$ is pure; that is, $\tau\kh{q}$ assigns probability one to a deterministic quitting time.

Adopting the same change of variables as in the proof of Proposition \ref{prop:opt}, we define  $\zeta(q)\equiv 1 - e^{-(r_p + \lambda)\tau(q)}\in[0,1]$, where $\zeta = 0$ corresponds to immediate quitting  ($\tau = 0$) and $\zeta=1$ corresponds to never quitting ($\tau=\infty$). With this transformation, the principal's and agent's values become
\eqns{W(\zeta) &= \mathbb{E}_{\bm\sigma}[w(\tau,q)] = \int_{0}^{\infty} w(q)\zeta(q) \, dF(q),\\
V(\zeta) &= \mathbb{E}_{\bm\sigma}[v(\tau,q)] = \int_{0}^{\infty} v(q)\phi(\zeta(q)) \, dF(q),}
where $\phi(z)\equiv 1-(1-z)^{1/\rho}$ and $\rho\equiv \frac{r_p+\lambda}{r+\lambda}$ as before.

To find the optimal information policy, we adopt the following approach: we first solve a relaxed problem that temporarily omits all interim incentive compatibility constraints,
\eqn{\begin{split}\max\limits_{\zeta\in:\kh{0,\infty}\to\fkh{0,1}}\quad& W(\zeta),\\
\text{s.t.}\quad\,\,\,\,\quad &V(\zeta)\ge 0,\\
\end{split}\label{eqn:relaxprog1}}
where the principal chooses the optimal $\zeta$ subject only to individual rationality. We then verify that the solution to \eqref{eqn:relaxprog1} can be implemented by an information policy, thereby establishing its optimality.

We set up the Lagrangian functional $\mathcal{L}$ with a multiplier $\eta \ge 0$ for the individual rationality constraint: $$\mathcal{L}\kh{\zeta, \eta} \equiv\int_0^\infty \left( w(q) \zeta(q) + \eta  v(q) \phi(\zeta(q)) \right)\,dF(q).$$
By separability, we can maximize the Lagrangian pointwise for each $q$:
\eqn{\max_{\zeta\in[0,1]}H\kh{z,q,\eta}\equiv w(q) z+ \eta  v(q) \phi(z).\label{eqn:pointwise}}
For an individually rational task $q\ge \overline{q}$ where $v(q) \ge 0$, both terms are increasing in $z$, so $z=1$ is optimal. Henceforth we focus on tasks that are not individually rational,  $q\in\kh{0,\overline{q}}$, where we can write $H\kh{z,q,\eta}=w(q) z- \eta \abs{v\kh{q}} \phi(z)$.

\begin{proof}[Proof of Proposition~\ref{prop:gKG}]
Suppose $r_p\le r$, then $\rho\le1$ and $\phi(z)$ is concave, implying that $H\kh{z,q,\eta}=w(q) z- \eta \abs{v\kh{q}} \phi(z)$ is convex in $z$ for all $q\in\kh{0,\overline{q}}$. Since a convex function on $[0,1]$ is maximized at the boundaries, the optimal solution is bang-bang: either $z=0$ or $z=1$. As a result,
$$\max_{z\in[0,1]}H\kh{z,q,\eta}=\max\hkh{w(q)- \eta \abs{v\kh{q}},0}.$$
Since $b\kh{q}$ is nondecreasing and  $\abs{v\kh{q}}$  is strictly decreasing, $w\kh{q}=\frac{\lambda{b\kh{q}}}{r_p+\lambda}$ is nondecreasing and the difference $w(q)- \eta \abs{v\kh{q}}$ is strictly increasing. Define $y\equiv\inf\hkh{q\in\kh{0,\overline{q}}:w(q)\ge \eta \abs{v\kh{q}}}$. It follows that $\max\hkh{w(q)- \eta \abs{v\kh{q}},0}=\kh{w(q)- \eta \abs{v\kh{q}}}\cdot \indic{q\ge y}$. Thus, the optimal $\zeta$ to \eqref{eqn:relaxprog1} is the step function $\zeta(q)=\indic{q\ge y}$. This is precisely the structure of KG: $\zeta(q)\in\{0, 1\}$ with a cutoff $y$ that exactly binds individual rationality, implying that KG is optimal among all dynamic information policies.
\end{proof}

\begin{proof}[Proof of Lemma~\ref{lem:DGD}]
Denote by $\tau$ the agent's quitting schedule under a cutoff cascade policy $\bm{\sigma}$ with cutoff quality function $\tilde{q}\kh{s}$. At any instant $s\ge 0$, if the task is incomplete and the agent has not been informed to quit, he infers that  $q\ge \tilde{q}\kh{s}$.  Thus, $\tau\kh{q}$ is the inverse of $\tilde{q}\kh{s}$:
$$q=\tilde{q}\kh{s}\quad\iff\quad s=\tau\kh{q}.$$
If the individual rationality constraint is violated, i.e., $\mathbb{E}_{\bm\sigma}[v(\tau,q)]<0$,  the agent ignores all information and quits immediately. Therefore, individual rationality is necessary for the agent to follow $\bm{\sigma}$.

Now suppose $\mathbb{E}_{\bm\sigma}[v(\tau,q)]\ge 0$.  We show that the agent's best response is to follow $\bm{\sigma}$ at any time $s\ge 0$ while the task remains incomplete. If the agent is informed that $q<\tilde{q}\kh{s}$, his expected continuation value is bounded above by $v\kh{\tilde{q}\kh{s}}\le v\kh{\overline{q}}=0$, so quitting is optimal. If the agent is informed that $q\ge\tilde{q}\kh{s}$, his expected continuation value promised by $\bm\sigma$ is $\mathbb{E}_{\bm\sigma}[v(\tau-s, q) \mid q \ge \tilde{q}(s)]$. We need to show that this value is nonnegative, thereby establishing that it is optimal for the agent to continue when he is informed to do so.

Define 
$$h_1(s) \equiv \mathbb{E}_{\bm\sigma}[v(\tau-s,q) \mid q\ge \tilde{q}(s)] \cdot \mathbb{P}_{\bm\sigma}\left(q\ge \tilde{q}(s)\right) = \int_{\tilde{q}(s)}^{\infty} v(\tau(q)-s,q)\,dF(q).$$
Since $\tau$ is the inverse of $\tilde{q}$, we have $\tau\kh{\tilde{q}\kh{s}}=s$. Using the Leibniz integral rule and the fact that $v(0, q)=0$, we obtain
$$h_1'(s) = \int_{\tilde{q}(s)}^{\infty} \left( -\frac{\partial v}{\partial t}(\tau(q)-s,q) \right) \,dF(q).$$
Since $v(t, q)=\left(1-e^{-(r+\lambda) t}\right) v(q)$, we have $\frac{\partial v}{\partial t}(t, q)=(r+\lambda) e^{-(r+\lambda) t} v(q)$. For $q\le\overline{q}$, we have $v\kh{q}\le 0$, which implies that $-\frac{\p v}{\p t}\kh{\tau\kh{q}-s,q}\ge 0$. For $q>\overline{q}$,  the agent is never asked to quit (as $\tilde{q}(s) \in [0, \overline{q}]$), so $\tau(q)=\infty$ and consequently $-\frac{\partial v}{\partial t}(\tau(q)-s, q) = 0$. Combining both cases, we have $-\frac{\p v}{\p t}\kh{\tau\kh{q}-s,q}\ge 0$ for all $q$, implying that $h_1'(s) \ge 0$. Since $h(0) = \mathbb{E}_{\bm\sigma}[v(\tau, q)] \ge 0$ by assumption, we conclude that $h_1(s) \ge 0$ for all $s \ge 0$. This completes the proof that the agent will follow $\bm{\sigma}$ at each time $s \ge 0$.
\end{proof}

\begin{proof}[Proof of Propositions \ref{prop:gstar} and \ref{prop:gFDD}]
When the principal is less patient than the agent ($r_p > r$), we have $\rho > 1$ and $\phi(z)$ is strictly convex. This implies that $H\kh{z,q,\eta}=w(q) z- \eta \abs{v\kh{q}} \phi(z)$ is strictly concave in $z$ for all $q\in\kh{0,\overline{q}}$. Any interior optimal solution to \eqref{eqn:pointwise} is determined by the first-order condition:
 $$\frac{\partial H}{\partial z} = w(q) - \eta \abs{v\kh{q}} \phi'(z) = 0.$$
Substituting $\phi'(z)=\frac{1}{\rho} (1-z)^{1/\rho -1}$ and rearranging yields
\begin{equation}
(1 - \zeta(q))^{1/\rho -1} = \frac{\rho w(q)}{\eta |v(q)|}=\frac{\rho}{\eta u(q)}.\label{eqn:FOC}
\end{equation}
We now examine the boundary conditions for the interior solution. Since $1 - \zeta(q) \in[0,1]$  and  $1/\rho -1<0$, the left-hand side of equation \eqref{eqn:FOC} is bounded below by $1$, which requires that $\rho\ge \eta u(q)$. Conversely, if $\rho< \eta u(q)$, the first-order condition would imply $\zeta(q)<0$.
The KKT conditions then imply the corner solution $\zeta(q) = 0$.

As $q$ increases from zero to $\overline{q}$, the ratio $u\kh{q}$ is strictly decreasing and ranges from the positive value $u\kh{0}$ to zero. Consequently, the analysis distinguishes between two cases, depending on whether $\rho \ge \eta u(0)$.
\\\\\noindent \textbf{Case 1:} If $\rho < \eta u(0)$,  then by the intermediate value theorem, there exists a unique $q^{**}\in (0, \overline{q})$ satisfying $\rho = \eta u(q^{**})$. For all $q < q^{**}$, we have $\rho < \eta u(q)$, so $\zeta(q) = 0$. For all $q \in \left[q^{**}, \overline{q}\right)$, the solution is interior and characterized by equation \eqref{eqn:FOC}.

Substituting $\zeta(q)=1 - e^{-(r_p + \lambda)\tau(q)}$ and $\rho= \frac{r_p+\lambda}{r+\lambda}$ into equation \eqref{eqn:FOC}  and simplifying, we obtain
$$e^{\left(r_p-r\right)\tau(q)}=\frac{\rho}{\eta u(q)}=\frac{u(q^{**})}{u(q)}\quad\implies\quad\tau(q)=\frac{1}{r_p-r}\log\kh{\frac{u\kh{q^{**}}}{u\kh{q}}}.$$
This is precisely the structure of OIGD: a cutoff cascade policy $\operatorname{IGD}\kh{q^{**}}$ with initial cutoff quality $q^{**}$ that exactly binds individual rationality. By Lemma \ref{lem:DGD}, the quitting time $\tau^{\operatorname{OIGD}}$ is indeed induced by the information policy OIGD, which establishes that OIGD is optimal among all dynamic information policies. This completes the proof of Proposition \ref{prop:gstar}.
\\\\\noindent \textbf{Case 2:} If $\rho \ge \eta u(0)$, the monotonicity of $u(q)$ ensures that $\rho \ge \eta u(q)$ for all $q \in (0, \overline{q})$. Consequently, the solution is interior for all $q \in (0, \overline{q})$ and characterized by equation \eqref{eqn:FOC}.

Substituting $\zeta(q)=1 - e^{-(r_p + \lambda)\tau(q)}$ and $\rho= \frac{r_p+\lambda}{r+\lambda}$ into equation \eqref{eqn:FOC}  and simplifying, we obtain
$$e^{\left(r_p-r\right)\tau(q)}=\frac{\rho}{\eta u(q)}=\frac{\rho}{\eta u(0)}\cdot\frac{u(0)}{u(q)}\quad\implies\quad\tau(q)=\frac{1}{r_p-r}\left(\log\kh{\frac{\rho}{\eta u(0)}}+\log\kh{\frac{u\kh{0}}{u\kh{q}}}\right).$$
Let $\tilde{t}\equiv\frac{1}{r_p-r}\log\kh{\frac{\rho}{\eta u(0)}}$. This is precisely the structure of MDGD: a cutoff cascade policy $\operatorname{DGD}\kh{\tilde{t}}$ with delay time $\tilde{t}$ that exactly binds individual rationality. By Lemma \ref{lem:DGD}, the quitting time $\tau^{\operatorname{MDGD}}$ is indeed induced by the information policy MDGD, which establishes that MDGD is optimal among all dynamic information policies. This completes the proof of Proposition \ref{prop:gFDD}.
\end{proof}

\subsection{Proofs for Section \ref{sec:Poi}}\label{app:Poi}
\begin{proof}[Proof of Lemma \ref{lem:tbar}]
Let $\tau$ be the agent's planned quitting time. His expected payoff $h(\tau)$ is given by $h_2(\tau) \equiv \mu v(\tau,H) + (1-\mu)v(\tau,L).$ Substituting the functional forms \eqref{eqn:v} leads to:
$$h_2(\tau) = \mu\left(1-e^{-(r+\lambda_{H}) \tau}\right)v(H) + (1-\mu)\left(1-e^{-(r+\lambda_{L}) \tau}\right)v(L).$$
To find the optimal $\tau$, consider the first-order derivative: \eqn{h_2'(\tau) = \mu (r+\lambda_H) e^{-(r+\lambda_{H}) \tau} v(H) + (1-\mu) (r+\lambda_L) e^{-(r+\lambda_{L}) \tau} v(L).\label{eqn:hprime}}

First assume $\mu\le\hat{\mu}$.  Since $\lambda_H>\lambda_L$, we have $e^{-\left(r+\lambda_{H}\right) \tau}<e^{-\left(r+\lambda_{L}\right) \tau}$  for any $\tau>0$. We can therefore bound the derivative $h_2'(\tau)$ as follows:
\eqns{h_2'\kh{\tau}&<\mu e^{-\left(r+\lambda_{L}\right) \tau}\kh{r+\lambda_H}v\kh{H}+\kh{1-\mu}e^{-\left(r+\lambda_{L}\right) \tau} \kh{r+\lambda_L}v\kh{L}\\
&=e^{-\left(r+\lambda_{L}\right) \tau}\kh{\mu\kh{r+\lambda_H}v\kh{H}+\kh{1-\mu}\kh{r+\lambda_L}v\kh{L}}\le 0,}
where the final inequality holds due to the definition of $\hat{\mu}$ in equation \eqref{eqn:muhat}. This shows that $h_2\kh{\tau}$ is strictly decreasing in $\tau$, and hence the optimal quitting time is $\tau=0$.

Now assume $\mu>\hat{\mu}$. We argue that $h_2\kh{\tau}$ attains its maximum at $\tau=\overline{t}\kh{\mu}$. Using the expression for $h_2'(\tau)$ in equation \eqref{eqn:hprime} and the definition of $\hat{\mu}$ in equation \eqref{eqn:muhat}, we have:
\eqns{h_2'\kh{\tau}>0\quad\iff\quad& e^{\kh{\lambda_H-\lambda_L}\tau}<\frac{\mu\kh{r+\lambda_H} v\kh{H}}{\kh{1-\mu}\kh{r+\lambda_L}\fkh{-v\kh{L}}}=\frac{\mu}{1-\mu}\cdot\frac{1-\hat{\mu}}{\hat{\mu}}\\
\quad\iff\quad&\tau< \frac{1}{\lambda_{H}-\lambda_{L}}\left(\log \frac{\mu}{1-\mu}-\log \frac{\hat{\mu}}{1-\hat{\mu}}\right)=\overline{t}\kh{\mu},}
where the last equality follows from the definition of $\overline{t}\kh{\mu}$ in equation \eqref{eqn:tbar}.
Therefore, $h_2(\tau)$ is strictly increasing on $\kh{0,\overline{t}\kh{\mu}}$, attains its maximum at $\tau=\overline{t}\kh{\mu}$, and is strictly decreasing on $\kh{\overline{t}\kh{\mu},\infty}$. The optimal strategy for the agent is to keep working until time $\overline{t}(\mu)$ and then quit if the task is not yet completed.
\end{proof}

\begin{proof}[Proof of Lemma \ref{lem:free}]
For any information policy $\bm{\sigma}$, let $\alpha$ denote the CDF of $\tau\kh{H}$, the agent's quitting time conditional on the task being high quality, and let $\beta$ denote the CDF of $\tau\kh{L}$. 

Assume $\mu>\hat{\mu}$. Consider an optimal information policy $\bm{\sigma}$. Since the agent never needs encouragement to continue when the task quality is high, we have $\alpha\kh{t}=0$ for all $t\in[0,\infty)$. To complete the proof, it remains to show that  $\beta\kh{t}=0$ for all $t\in\left[0,\overline{t}\kh{\mu}\right)$.

Suppose to the contrary that $\beta\kh{t}>0$ for some  $t\in\left[0,\overline{t}\kh{\mu}\right)$. Consider the perturbation $\tilde{\beta}$ given by
$$\tilde{\beta}\kh{t}=\Brace{&0,&&t\in\left[0,\overline{t}\kh{\mu}\right),\\
&\beta\kh{t},&&t\in\left[\overline{t}\kh{\mu},\infty\right).}$$
Note that $\tilde{\beta}$ first-order stochastically dominates $\beta$, as it shifts the mass accumulated on $\left[0,\overline{t}\kh{\mu}\right)$ to a discrete jump at $t=\overline{t}\kh{\mu}$. The principal's payoff can be written as:
$$W(\beta)\equiv\mathbb{E}_{\bm\sigma}\fkh{w\kh{\tau,q}}=\mu w(H) + (1-\mu)w(L)\int_{0}^\infty\left(1-e^{-(r_{p}+\lambda_L) t}\right) \,d\beta(t).$$
Since the integrand $1-e^{-(r_p+\lambda_L)t}$ is strictly increasing in $t$,  the principal strictly prefers $\tilde{\beta}$ to $\beta$.  To complete the contradiction, it remains only to verify that $\tilde{\beta}$ satisfies the agent's incentive compatibility constraints.

For any time $t< \overline{t}\kh{\mu}$, under $\tilde{\beta}$, the principal provides no information. The agent's belief $\mu_t$ follows the natural decay path given by equation \eqref{eqn:mu_t}, and by definition of $\overline{t}\kh{\mu}$ in equation \eqref{eqn:tbar}, we have $\mu_t>\hat{\mu}$, so the agent is willing to work at all $t< \overline{t}\kh{\mu}$. At any time $t\ge \overline{t}\kh{\mu}$,  the conditional distribution of quitting times under $\tilde{\beta}$ is identical to that under $\beta$. Since $\beta$ was incentive compatible, the constraints at $t \ge \overline{t}\kh{\mu}$ remain satisfied. This contradicts the optimality of  $\beta$,  completing the proof.
\end{proof}

\begin{proof}[Proof of Proposition \ref{prop:optlearn}]
For $\mu> \hat{\mu}$, optimality follows directly from the Free-Riding Lemma (Lemma \ref{lem:free}). We therefore assume $\mu\le \hat{\mu}$ for the remainder of the proof.

Since the agent never needs encouragement to continue when the task quality is high, the non-trivial design problem concerns only the low-quality task. Adopting the same change of variables as in the proof of Proposition \ref{prop:opt}, we define  $\zeta\equiv 1 - e^{-(r_p + \lambda_L)\tau(L)}\in[0,1]$, where $\zeta = 0$ corresponds to immediate quitting  ($\tau(L) = 0$) and $\zeta=1$ corresponds to never quitting ($\tau(L)=\infty$),  and let $G$ denote this distribution of $\zeta$. The principal's expected payoff is then
$$W(G) = \mathbb{E}_{\bm\sigma}[w(\tau,q)] =\mu w(H) + (1-\mu) w(L) \mathbb{E}_G [\zeta].$$
Since the first term is constant and $w(L) > 0$,  the principal's problem reduces to maximizing $\mathbb{E}_G [\zeta]$. 

The agent's expected payoff is 
$$V(G) = \mathbb{E}_{\bm\sigma}[v(\tau,q)] =\mu v(H) + (1-\mu) v(L) \mathbb{E}_G\fkh{\phi(\zeta)},$$
where $\phi(z)\equiv 1-(1-z)^{1/\rho}$ and $\rho\equiv \frac{r_p+\lambda_L}{r+\lambda_L}$. Since $v(L) < 0$, the individual rationality constraint $V(G)\ge 0$ rearranges to $\mathbb{E}_G\fkh{\phi(\zeta)}\le C_1$ where $C_1\equiv \frac{\mu v\kh{H}}{\kh{1-\mu}\fkh{-v\kh{L}}}$.

Next, we simplify the interim incentive compatibility constraints. At any time $s>0$, the agent's continuation value under $\bm\sigma$ can be expanded using the law of iterated expectations, conditioning on the realization of task quality $q$:
\begin{align*}
\mathbb{E}_{\bm\sigma}\fkh{v\kh{\tau-s,q}\mid x>s,\tau>s} &= \underbrace{\mathbb{P}_{\bm{\sigma}}\kh{q=H \mid x>s, \tau>s}}_{\mu_s} \cdot \mathbb{E}_{\bm\sigma}\fkh{v\kh{\tau-s,H} \mid x>s, \tau>s, q=H} \\
&\quad + \underbrace{\mathbb{P}_{\bm{\sigma}}\kh{q=L  \mid  x>s, \tau>s}}_{1-\mu_s} \cdot \mathbb{E}_{\bm\sigma}\fkh{v\kh{\tau-s,L} \mid x>s, \tau>s, q=L},
\end{align*}
where the posterior belief is $\mu_s$ is given by
\eqn{\mu_s\equiv  \mathbb{P}_{\bm{\sigma}}\kh{q=H  \mid  x>s, \tau>s} = \frac{\mu }{\mu+ (1-\mu) e^{\kh{\lambda_H-\lambda_L }s} \mathbb{P}_{\bm\sigma}\kh{\tau>s \mid q=L}}.\label{eqn:mus}}
Equation \eqref{eqn:mus} is derived from Bayes' rule and nests equation \eqref{eqn:mu_t} as the special case with no disclosure. 

When the task quality is high, the agent never quits, so $\mathbb{E}_{\bm\sigma}\fkh{v\kh{\tau-s,H} \mid x>s, \tau>s, q=H}=v(H)$. When the task quality is low, the event $\tau>s$ corresponds to $\zeta>z_s$ in the transformed domain, where $z_s\equiv 1 - e^{-(r_p + \lambda_L)s}$. It follow that $\mathbb{P}_{\bm\sigma}\kh{\tau>s \mid q=L}=\mathbb{P}_{G}\kh{\zeta>z_s}$. The agent's continuation value at time $s$, $v\kh{\tau-s,L}$,  is proportional to $1 - e^{-(r+\lambda_L)(\tau-s)}$, which can be expressed as 
$$1 - e^{-(r+\lambda_L)(\tau-s)}=1 - \frac{e^{-(r+\lambda_L)\tau}}{e^{-(r+\lambda_L)s}} =1 - \frac{1-\phi(\zeta)}{1-\phi(z_s)}=\frac{\phi(\zeta) - \phi(z_s)}{1-\phi(z_s)}.$$ Therefore, $$\mathbb{E}_{\bm\sigma}\fkh{v\kh{\tau-s,L} \mid x>s, \tau>s, q=L} = v(L) \frac{\mathbb{E}_G\fkh{\phi(\zeta) - \phi(z_s)\mid \zeta>z_s}}{1-\phi(z_s)}.$$

Substituting these values and the belief $\mu_s$ from equation \eqref{eqn:mus} into the interim incentive compatibility constraint at time $s>0$,  $\mathbb{E}_{\bm\sigma}\fkh{v\kh{\tau-s,q}\mid x>s,\tau>s}\ge 0$, and clearing the strictly positive denominator, yields:
$$\mu v(H) + (1-\mu) e^{(\lambda_H-\lambda_L)s} \mathbb{P}_{G}\kh{\zeta>z_s} v(L)\frac{\mathbb{E}_G\fkh{\phi(\zeta) - \phi(z_s)\mid \zeta>z_s}}{1-\phi(z_s)} \ge 0.$$
Combining the probability $\mathbb{P}_{G}(\zeta>z_s)$ with the conditional expectation:
$$\mu v(H) + (1-\mu) e^{(\lambda_H-\lambda_L)s} v(L) \frac{\mathbb{E}_G\fkh{\left(\phi(\zeta) - \phi(z_s)\right)\indic{\zeta>z_s}}}{1-\phi(z_s)} \ge 0.$$
Since $v(L) < 0$, rearranging the inequality to isolate the expectation term yields:
\eqns{\mathbb{E}_G\fkh{\left(\phi(\zeta) - \phi(z_s)\right)\indic{\zeta>z_s}}\le \frac{\mu v\kh{H}}{\kh{1-\mu}\fkh{-v\kh{L}}}\cdot e^{-(\lambda_H-\lambda_L)s}\kh{1-\phi(z_s)}=C_1e^{-(\lambda_H-\lambda_L)s}\kh{1-\phi(z_s)}.}
Expressing the right-hand side in terms of $z_s$ using the identities $e^{-(\lambda_H-\lambda_L)s} = \left( e^{-(r_p+\lambda_L)s} \right)^{\frac{\lambda_H-\lambda_L}{r_p+\lambda_L}} = (1-z_s)^{\frac{\lambda_H-\lambda_L}{r_p+\lambda_L}}$ and $1-\phi(z_s) = (1-z_s)^{1/\rho} = (1-z_s)^{\frac{r+\lambda_L}{r_p+\lambda_L}}$:
$$\mathbb{E}_G\fkh{\left(\phi(\zeta) - \phi(z_s)\right)\indic{\zeta>z_s}}\le C_1(1-z_s)^{\frac{\lambda_H-\lambda_L}{r_p+\lambda_L}}(1-z_s)^{\frac{r+\lambda_L}{r_p+\lambda_L}} =C_1(1-z_s)^{\frac{r+\lambda_H}{r_p+\lambda_L}}.$$
The transformation $z_s = 1 - e^{-(r_p+\lambda_L)s}$ maps $s \in (0,\infty)$ bijectively to $z \in (0,1)$; thus, satisfying the constraint for all $s > 0$ is equivalent to satisfying it for all $z \in (0,1)$. The boundary condition at $z = 0$ represents the individual rationality constraint $\mathbb{E}_G\left[\phi(\zeta)\right] \le C_1$, while the condition at $z = 1$ is satisfied trivially since $\zeta \le 1$. Therefore, to find the optimal information policy, we can rewrite the principal's program as
\eqn{\begin{split}\max\limits_{G\in\Delta\kh{\fkh{0,1}}}\quad& \mathbb{E}_G [\zeta],\\
\text{s.t.}\quad\quad &\mathbb{E}_G\fkh{\left(\phi(\zeta) - \phi(z)\right)\indic{\zeta>z}}\le C_1(1-z)^{\frac{r+\lambda_H}{r_p+\lambda_L}},\quad \forall z\in[0,1],\\
\end{split}\label{eqn:reformprogl}}

If $r_p<r$, then $\rho< 1$ and $\phi(z)$ is strictly concave. Let 
$$h_3(z) \equiv \mathbb{E}_G\fkh{\left(\phi(\zeta) - \phi(z)\right)\indic{\zeta>z}} = \int_z^1 \left( \phi(\zeta) - \phi(z) \right)\,dG(\zeta).$$
By integration by parts,
$$h_3(z) = \left.(\phi(\zeta) - \phi(z))(G(\zeta) - 1) \right|_{\zeta=z}^{\zeta=1} - \int_z^1 (G(\zeta) - 1) \phi'(\zeta)\,d\zeta =\int_z^1 (1-G(\zeta)) \phi'(\zeta)\,d\zeta.$$
Differentiating with respect to $z$ using the Leibniz integral rule yields
\begin{equation}
h_3'(z) = -(1-G(z)) \phi'(z).\label{eqn:uprime}
\end{equation}
The principal's problem \eqref{eqn:reformprogl} is to maximize $\mathbb{E}_G[\zeta] = \int_0^1 (1-G(z)) \,dz$ subject to $h_3(z) \le C_1(1-z)^{\frac{r+\lambda_H}{r_p+\lambda_L}}$ for all $z\in[0,1]$. Using equation \eqref{eqn:uprime}, we can express the objective in terms of $h_3(z)$: $$\int_0^1 (1-G(z)) dz = \int_0^1 \frac{-h_3'(z)}{\phi'(z)}\,dz.$$
Apply integration by parts again, and note that $h_3(1)=0$:
$$\mathbb{E}_G[\zeta]=\int_0^1 \frac{-h_3'(z)}{\phi'(z)}\,dz = \left.\frac{-h_3(z)}{\phi'(z)} \right|_0^1 - \int_0^1 h_3(z) \frac{d}{dz}\left( \frac{-1}{\phi'(z)} \right)\,dz=\frac{h_3(0)}{\phi'(0)} + \int_0^1 h_3(z){\left(\frac{-\phi''(z)}{\phi'(z)^2} \right)}\,dz.$$
Since $\phi(z)$ is strictly concave, $-\phi''(z)>0$, so the objective in \eqref{eqn:reformprogl} is strictly increasing in $h_3(z)$ pointwise. To maximize the objective, the principal must set $h_3(z)$ to its maximum feasible value, which means the constraint binds everywhere: $h_3(z) = C_1(1-z)^{\frac{r+\lambda_H}{r_p+\lambda_L}}$ for all $z\in[0,1]$. Since $\phi^{\prime}(z)= \frac{1}{\rho}(1-z)^{1/\rho -1}= \frac{r+\lambda_L}{r_p+\lambda_L} (1-z)^{\frac{r+\lambda_L}{r_p+\lambda_L}-1}$, we can solve for $1-G(z)$ using equation \eqref{eqn:uprime}:
$$1-G(z)=\frac{-h_3'(z)}{\phi'(z)}=\frac{C_1\cdot\frac{r+\lambda_H}{r_p+\lambda_L} \cdot(1-z)^{\frac{r+\lambda_H}{r_p+\lambda_L}-1}}{\frac{r+\lambda_L}{r_p+\lambda_L}\cdot(1-z)^{\frac{r+\lambda_L}{r_p+\lambda_L}-1}}=C_{1} \cdot \frac{r+\lambda_{H}}{r+\lambda_{L}} \cdot(1-z)^{\frac{\lambda_{H}-\lambda_{L}}{r_{p}+\lambda_{L}}}.$$
From the definition of $\hat{\mu}$ in equation \eqref{eqn:muhat}, the condition $\mu\le \hat{\mu}$ is equivalent to $C_1=\frac{\mu v\kh{H}}{\kh{1-\mu}\fkh{-v\kh{L}}}\le \frac{r+\lambda_L}{r+\lambda_H}$.
This ensures that $1-G(0) \le 1$, confirming that the solution constitutes a valid probability distribution.
Substituting $\zeta=1 - e^{-(r_p + \lambda_L)\tau}$, which implies that $(1-\zeta)^{\frac{\lambda_{H}-\lambda_{L}}{r_{p}+\lambda_{L}}}=e^{-(\lambda_{H}-\lambda_{L})\tau}$, yields
$$\mathbb{P}_{\bm\sigma}\kh{\tau>s \mid q=L}=C_{1} \cdot \frac{r+\lambda_{H}}{r+\lambda_{L}} \cdot e^{-(\lambda_{H}-\lambda_{L})s}$$
This is precisely the structure of IPD: immediate disclosure with probability $1-C_{1} \cdot \frac{r+\lambda_{H}}{r+\lambda_{L}}$ followed by Poisson disclosure at rate $\lambda_H-\lambda_L$, implying that IPD is optimal among all dynamic information policies.

If $r_p>r$, then $\rho>1$ and $\phi(z)$ is strictly convex. By Jensen's inequality, $\mathbb{E}_G [\phi(z)] \ge \phi(\mathbb{E}_G [z])$, and equality holds if and only if $G$ is degenerate  at a specific point. It follows that $\mathbb{E}_G [z]\le \phi^{-1}\kh{\mathbb{E}_G [\phi(z)] }\le \phi^{-1}\kh{C_1}$, with equality if and only if $G$  is degenerate at $z^*=\phi^{-1}\kh{C_1}$.  This is precisely the structure of MDD: a deterministic value $\zeta=z^*$ that exactly binds individual rationality. It remains to verify that the agent will follow the recommendations given by MDD; that is, the constraints in \eqref{eqn:reformprogl} hold for all $z$. For $z\ge z^*$, we have $\indic{\zeta>z}=0$, so the left-hand side equals zero and the constraint is trivially satisfied. For $z<z^*$, the constraint becomes: $\phi(z^*) - \phi(z)\le C_1(1-z)^{\frac{r+\lambda_H}{r_p+\lambda_L}}$. Since $\phi(z^*)=C_1$, this is equivalent to $ \phi(z)\ge C_1\kh{1-(1-z)^{\frac{r+\lambda_H}{r_p+\lambda_L}}}$. Note that $\phi(z)=1-(1-z)^{1/\rho}$,  so the inequality holds trivially at $z=0$. Let $\kappa\equiv \frac{r_p+\lambda_L}{r+\lambda_H}$. For $z\in(0,z^*)$, it suffices to show
\eqn{h_4(z)\equiv\frac{1-(1-z)^{1/\rho}}{1-(1-z)^{1/\kappa } } \ge C_1.\label{eqn:h4}}
Since $\lambda_H>\lambda_L$, the exponents satisfy $1/\rho< 1/\kappa$. We first show that $h_4(z)$ is strictly increasing. Differentiating with respect to $z$ yields:$$h_4'(z) = \frac{\frac{1}{\rho}(1-z)^{1/\rho-1}\left(1-(1-z)^{1/\kappa}\right) - \frac{1}{\kappa}(1-z)^{1/\kappa-1}\left(1-(1-z)^{1/\rho}\right)}{\left(1-(1-z)^{1/\kappa}\right)^2}.$$
The sign of $h_4'(z)$ is determined by the numerator. Factoring out $(1-z)^{1/\rho-1}$, we define the remaining term:
\eqns{h_5(z) &\equiv \frac{1}{\rho}\left[1-(1-z)^{1/\kappa}\right] - \frac{1}{\kappa}(1-z)^{1/\kappa-1/\rho}\left[1-(1-z)^{1/\rho}\right]\\
&=\frac{1}{\rho} - \frac{1}{\kappa}(1-z)^{1/\kappa-1/\rho} +\left(\frac{1}{\kappa} - \frac{1}{\rho}\right)(1-z)^{1/\kappa}.}
We have $h_5(0)=\frac{1}{\rho} - \frac{1}{\kappa} +\left(\frac{1}{\kappa} - \frac{1}{\rho}\right)=0$, and
\eqns{h_5'(z) &= \frac{1}{\kappa}\left(\frac{1}{\kappa}-\frac{1}{\rho}\right)(1-z)^{1/\kappa-1/\rho-1} - \left(\frac{1}{\kappa} - \frac{1}{\rho}\right)\frac{1}{\kappa}(1-z)^{1/\kappa-1}\\
&= \frac{1}{\kappa}\left(\frac{1}{\kappa}-\frac{1}{\rho}\right)(1-z)^{1/\kappa-1/\rho-1} \left(1 - (1-z)^{1/\rho}\right)>0.}
Thus, $h_5(z) > 0$ for all $z \in (0, 1)$, implying $h_4'(z) > 0$. Since $h_4(z)$ is strictly increasing, the inequality \eqref{eqn:h4} is most restrictive as $z \to 0^+$. Applying l'H\^{o}pital's rule:$$\lim_{z \to 0} h_4(z) = \lim_{z \to 0} \frac{\frac{1}{\rho}(1-z)^{1/\rho-1}}{\frac{1}{\kappa}(1-z)^{1/\kappa-1}}=\frac{\kappa}{\rho} = \frac{r+\lambda_L}{r+\lambda_H}.$$Therefore, $h_4(z) \ge C_1$ holds for all $z \in (0, z^*)$ if and only if $\frac{r+\lambda_L}{r+\lambda_H} \ge C_1$. From the definition of $\hat{\mu}$ in equation \eqref{eqn:muhat}, this is equivalent to $\mu \le \hat{\mu}$, which is satisfied by assumption. Thus, the constraints in \eqref{eqn:reformprogl} hold for all $z$, confirming that the agent will follow the recommendations given by MDD. Therefore, MDD is optimal among all dynamic information policies.

If $r_p=r$, then $\rho=1$ and $\phi(z)=z$, so the objective and constraint in \eqref{eqn:reformprogl} have the same functional form. Thus, any policy that satisfies the interim incentive compatibility constraints while exactly binding the individual rationality constraint is optimal, including both IPD and MDD.
\end{proof}
\end{document}